\newif\ifarxiv
\arxivtrue
\newif\ifjournal
\journalfalse

\ifjournal
\documentclass{iopart}
\else
\documentclass{article}
\fi

\ifarxiv
\usepackage{times}
\usepackage{savesym}
\usepackage{amsmath,amsthm,amsbsy}
\savesymbol{iint}
\savesymbol{openbox}
\usepackage{txfonts}
\restoresymbol{TXF}{iint}
\restoresymbol{TXF}{openbox}
\numberwithin{equation}{section}
\else
\ifjournal
\usepackage{amsopn,amsbsy,amsthm}
\def\eqref#1{(\ref{#1})}
\eqnobysec
\else
\usepackage[lf]{MinionPro}
\usepackage{amsmath,amsthm,amsbsy}
\numberwithin{equation}{section}
\fi
\fi

\usepackage{geometry}
\ifjournal
\else
\usepackage{color}
\definecolor{labelkey}{rgb}{0,0,.75}
\definecolor{MyGreen}{rgb}{0,.6,.2}
\definecolor{MyDarkBlue}{rgb}{.1,.1,.75}
\usepackage[colorlinks=true, citecolor=MyGreen, linkcolor=MyDarkBlue]{hyperref}
\fi
\usepackage{tensor}
\usepackage{xypic}

\ifjournal
\makeatletter
\newcommand{\customlabel}[2]{%
\protected@write \@auxout {}{\string \newlabel {#1}{{#2}{}{}}}}
\makeatother
\fi

\date{June 4, 2014}

\DeclareMathOperator{\ck}{\bf L}
\DeclareMathOperator{\Lap}{\Delta}
\let\tr\relax
\DeclareMathOperator{\tr}{\rm tr}
\DeclareMathOperator{\Lie}{\mathcal{L}}

\def\ip<#1,#2>{\left<#1,#2\right>}
\renewcommand{\div}{\mathop{\rm div}\nolimits}

\newcommand{\ra}{\rightarrow}

\let\Im\relax
\DeclareMathOperator{\Im}{\rm Im}
\let\ovl\overline

\newcommand{\calL}{\mathcal{L}}
\newcommand{\calN}{\mathcal{N}}
\newcommand{\calM}{\mathcal{M}}

\newcommand{\calD}{\mathcal{D}}

\newcommand{\calC}{\mathcal{C}}

\newcommand{\calX}{\mathcal{X}}
\newcommand{\calY}{\mathcal{Y}}

\newcommand{\bfsigma}{\boldsymbol{\sigma}}

\begin{document}
\title{The Conformal Method and the Conformal Thin-Sandwich Method Are the Same}
\author{David Maxwell}

\ifjournal
\newtheorem{theorem}{Theorem}[section]
\newtheorem{proposition}[theorem]{Proposition}
\newtheorem{definition}{Definition}
\newtheorem{lemma}[theorem]{Lemma}
\else
\newtheorem{theorem}{Theorem}[section]
\newtheorem{conjecture}[theorem]{Conjecture}
\newtheorem{problem}{Problem}
\newtheorem{definition}{Definition}
\newtheorem{proposition}[theorem]{Proposition}
\newtheorem{corollary}[theorem]{Corollary}
\newtheorem{lemma}[theorem]{Lemma}
\fi

\maketitle

\begin{abstract}
The conformal method developed in the 1970s and
the more recent Lagrangian and Hamiltonian conformal thin-sandwich methods
are techniques for finding solutions of the Einstein constraint equations.
We show that they are manifestations of a single conformal method: there is a 
straightforward way to convert back and forth between the parameters 
for these methods so that the corresponding solutions of the Einstein
constraint equations agree.  The unifying idea is the need
to clearly distinguish tangent and cotangent vectors to the space of
conformal classes on a manifold, and we introduce a vocabulary
for working with these objects without reference to a particular
representative background metric. As a consequence of these conceptual
advantages, we demonstrate how to strengthen previous near-CMC existence
and non-existence theorems for the original conformal method to include
metrics with scalar curvatures that change sign.
\end{abstract}

\section{Introduction}

Initial data for the (vacuum) Cauchy problem in general relativity consist of
a Riemannian manifold $(M^n, g_{ab})$ and a symmetric tensor $K_{ab}$
that will become the induced metric and second fundamental form
of an embedding of $M^n$ into a Ricci flat Lorentzian spacetime determined
from the initial data $(g_{ab},K_{ab})$. The Gauss and Codazzi relations (along
with the fact that the ambient spacetime is Ricci flat) impose the following
compatibility conditions on the Cauchy data:
\ifjournal
\numparts
\begin{eqnarray}\label{eq:constraints}
R_{g} - |K|_{g}^2 + (\tr_g K)^2 = 0 &\qquad\text{\small[Hamiltonian constraint]}\label{eq:hamiltonian}\\
\div_{g} K = d\tau&\qquad\text{\small[momentum constraint]}\label{eq:momentum}
\end{eqnarray}
\endnumparts
\else
\begin{subequations}\label{eq:constraints}
\begin{alignat}{2}
R_{g} - |K|_{g}^2 + (\tr_g K)^2 &= 0 &\qquad&\text{\small[Hamiltonian constraint]}\label{eq:hamiltonian}\\
\div_{g} K &= d\tau&\qquad&\text{\small[momentum constraint]}\label{eq:momentum}
\end{alignat}
\end{subequations}\fi
where $\tau = g^{ab}K_{ab}$ is the mean curvature.  
Choquet-Bruhat showed \cite{FouresBruhat:1952wg} that there exists a solution of the Cauchy problem 
if and only if the initial data satisfy the Einstein constraint equations 
\ifjournal\eqref{eq:hamiltonian}-\eqref{eq:momentum}\else\eqref{eq:constraints}\fi,
so finding solutions of the constraint equations is a fundamental problem in general relativity.

In 1944, Lichnerowicz \cite{Lichnerowicz:1944} initiated an approach for finding solutions of the constraint
equations, the so-called conformal method.  Extensions of this method (due to York and his collaborators, 
as described below)
are now the principal techniques used to construct solutions from scratch. For these methods, one starts with a 
Riemannian metric $g_{ab}$ 
and attempts to construct a solution $(\ovl g_{ab}, \ovl K_{ab})$ of the constraints where
$\ovl g_{ab}$ is conformally related to $g_{ab}$ via some conformal factor to be determined as part of the solution.
The mean curvature $\tau = \ovl g^{ab} \ovl K_{ab}$ is also freely specified so
\begin{equation}
\ovl K_{ab} = \ovl A_{ab} + \frac{2}{n}\tau\,\ovl g_{ab}
\end{equation}
where $\ovl A_{ab}$ is a trace-free tensor that is determined, as part of the solution procedure, 
from $g_{ab}$, $\tau$, and other auxiliary data specific to the particular conformal method.

Lichnerowicz's original conformal method constructed solutions with $\tau \equiv 0$.
In the early '70s York proposed an extension of the method that
allows one to specify $\tau = \tau_0$ for an arbitrary constant $\tau_0$ \cite{York:1973fl}, and
with \'O Murchadha described a further extension to arbitrary mean curvatures
\cite{OMurchadha:1974bf}.
In this paper, we refer to the prescription from \cite{York:1973fl}
as the CMC conformal method, and its extension in \cite{OMurchadha:1974bf} as 
the 1974 conformal method. The data for this method consist of a metric $g_{ab}$, 
a symmetric trace-free tensor $\sigma_{ab}$ satisfying $\nabla^a \sigma_{ab}=0$, 
and a mean curvature $\tau$.

Somewhat later York proposed a Lagrangian conformal thin-sandwich (CTS) approach \cite{YorkJr:1999jo},
and subsequently with Pfeiffer described a Hamiltonian formulation of the CTS method 
\cite{Pfeiffer:2003ka}.\footnote{We distinguish here between the CTS method and the so-called \textit{extended} 
CTS method described at the end of \cite{Pfeiffer:2003ka}, which is a nontrivial
modification of the CTS method and has unsatisfactory uniqueness properties \cite{Pfeiffer:2005iz}.  
We do not treat the extended CTS method.}
The Lagrangian method starts from initial data $(g_{ab},u_{ab},\tau, N)$ where $u_{ab}$ is an arbitrary 
trace-free symmetric tensor and $N$ is a positive function related to a parameter (the lapse) that appears in 
the $n$+1 formulation of the Cauchy problem.  The CTS method has the virtue that given
conformal data $(g_{ab}, u_{ab}, \tau, N)$ and a second conformally related metric $\tilde g_{ab}$,
there is a way to conformally transform the remaining data to form
$(\tilde g_{ab}, \tilde u_{ab}, \tau, \tilde N)$ such that the set of solutions of the constraints
associated with the original data and the transformed data are the same.  The property is 
known as conformal covariance (or sometimes conformal invariance), and is 
a property that is shared with the CMC conformal
method but that is apparently absent for the 1974 conformal method.  The ability
to select a background metric within the conformal class satisfying some desired property 
(e.g. a metric with a scalar curvature that has constant sign) is a powerful tool, and
this has occasionally lead to theorems that are stronger when using the CTS approach 
versus the 1974 conformal method.

The purpose of this note is to clarify the relationship between the CTS methods and the 
1974 conformal method: they are the same.
Specifically, there is a way to translate,
in a straightforward and essentially unique way, between 1974 conformal 
data $(g_{ab}, \sigma_{ab}, \tau)$ and 
CTS data $(\hat g_{ab}, \hat u_{ab}, \tau, \hat N)$ 
such that the corresponding solutions of the constraint equations are the same.
The significance of this result arises from the fact that 
if $\tau$ is not constant, then it is generally unknown
how many solutions are associated with 1974 conformal data 
$(g_{ab}, \sigma_{ab}, \tau)$.  One hopes that there is exactly one, except perhaps
for some well defined set of data where there is none.  But from
\cite{Maxwell:2011if} we have examples 
showing that there can be more than one solution, and evidence
that the set for which there is no solution may be difficult to describe.
Since the 1974 and CTS methods are the same, these deficiencies apply equally
to both methods.  Or, from a more positive perspective,
we see that any result that can be proved for one method can be translated into
an equivalent theorem about the other method.
Questions of which data yield no solutions, or exactly one solution, or multiple solutions of the constraints
can be formulated using whichever method is convenient.  Moreover, since the CTS method
is conformally covariant, \textit{so is} the 1974 conformal method when handled correctly.

A remark that the methods are identical (along with a sketch of the equivalence)
was made by the author in \cite{Maxwell:2011if}, which explored conformal parameterizations
of certain far-from CMC data using the CTS framework, and which asserted that the results 
of that paper translate to results for the 1974 conformal method.  Aside from the remark in 
\cite{Maxwell:2011if}, the fact that the methods are the same does
not seem to appear in the literature (although \cite{Pfeiffer:2003ka} comes
very close, but stops short and seems
to have the perspective that the methods are different).
Since the mathematics literature in recent years has seen progress toward understanding 
the 1974 conformal method for non-CMC conformal data 
(e.g. \cite{Holst:2009ce} \cite{Maxwell:2009co} 
\cite{Dahl:2012ef}), and since physicists tend to use the CTS method, 
it seems useful to have a guide 
for how to translate results between the methods.
  Moreover, there 
are instances in both mathematics
and physics publications where the methods are asserted to be different,
or where results are proved for the 1974 method that are weaker than analogous 
results for the CTS method.  The following examples illustrate  how improvements
can be realized by taking advantage of the equivalence.
\begin{itemize}
\item In \cite{Isenberg:1996fi} and
\cite{Allen:2008ef} the 1974 approach is used to generate near-CMC solutions of the 
constraint equations, but under the restriction that the scalar curvature
of the metric has constant sign.  These theorems admit generalizations
to the CTS method, and we will see in Section \ref{sec:apps} that
these can be used to establish similar near-CMC 
results for the 1974 method for metrics with an arbitrary scalar curvature.
\item Reference \cite{Isenberg:2004jd} contains nonexistence theorems for the conformal methods
for certain non-CMC data.
Theorem 2 (framed in the 1974 conformal framework) is weaker than Theorem 3 (which uses
the CTS framework) and it is asserted that the gap is related to the lack 
of conformal covariance of the 1974 method. However, the 1974 method is covariant and
we show in Section \ref{sec:apps} that Theorem 2 can be improved to be just as strong as Theorem 3.
\item In \cite{Dahl:2012ef}, Theorem 1.7 proves a variation of the nonexistence result of
\cite{Isenberg:2004jd} Theorem 2 using the 1974 conformal method approach; 
it can similarly be strengthened by taking advantage
of the equivalence of the 1974 and CTS methods.\footnote{We note in press
that the recent preprint \cite{Ahn:2014} provides an alternative proof that also strengthens \cite{Dahl:2012ef}
Theorem 1.7.}
\item The numerical relativity text \cite{Baumgarte:2010vs}
presents the 1974 conformal method (called there the conformal
transverse traceless decomposition) and the CTS method
as different techniques, with a different number of 
specifiable parameters for each method.  We show here how
to translate back and forth between the parameters of the two methods; 
knowing that the parameterizations are the same gives insight into both methods.
\end{itemize}

A secondary goal of this paper is to formulate the parameters of
the various conformal methods in terms of objects at the level of the set $\calC$ 
of conformal classes rather than the set $\calM$ of metrics on $M$.  I.e., 
we work with conformal classes,
tangent and cotangent vectors to $\calC$, and conformal Killing operators defined
in terms of conformal classes rather than representative metrics.  Doing so
can be thought of as a coordinate-free approach to understanding the parameters.
Motivated by diffeomorphism invariance of the Einstein equations, 
we also give a clear interpretation of these parameters as objects associated
with $\calC/\calD_0$, where $\calD_0$ is the connected component of the identity
of the diffeomorphism group.

While a coordinate-free perspective is implicitly present in some of
the physics literature, it clarifies matters to make it explicit.  
For example, it turns out that it is crucial to make a distinction between the 
tangent space $T_{\mathbf g} \calC$ and the cotangent space $T_{\mathbf g}^* \calC$ 
to $\calC$ at some conformal class $\mathbf{g}$.
Unlike the case for $\calM$, there is no natural way to identify tangent 
vectors as cotangent vectors, but there is a natural family of identifications.
In the CTS method, a choice from this family
is specified via the lapse.  In the 1974 method the choice
is specified less explicitly, and this is perhaps the reason why
it is not obvious at first glance that the 1974 and CTS methods are the same.
Most variations of the conformal method use the metric parameter
$g_{ab}$ simply to determine the conformal class $[g_{ab}]$
of the solution metric. For the 1974 conformal method, however,
the choice of $g_{ab}$ specifies both $[g_{ab}]$ and a
choice of identification of $T_{\mathbf{g}} \calC$ with
$T_{\mathbf{g}}^* \calC$.  Changing the 
representative metric in the 1974 approach is equivalent to changing the
lapse in the CTS approach, and working with the parameters in a coordinate-free way 
helps make this relationship clear.

Although the coordinate-free formulation provides insight, it also introduces
an extra layer of abstraction.  Readers who are already familiar with the conformal
methods, and who wish to skip over this abstraction, can jump to the end of Section \ref{sec:same}
where there are concise recipes, in familiar tensorial terms, for how to convert parameters
between the various conformal methods.  These recipes describe mechanically \textit{how} the methods are the same;
we hope that the coordinate-free approach taken elsewhere in the paper
illuminates \textit{why} the methods are the same.  

The remainder of the paper proceeds as follows.  
In Section \ref{sec:conformal} we establish coordinate-free language for describing
conformal objects, and Sections \ref{sec:CMC} through \ref{sec:CTSH} formulate
each of the various conformal methods in terms of this language.  In Section
\ref{sec:same} we establish the equivalence of all these methods, and
Section \ref{sec:apps} uses this equivalence in some applications.

\subsection{Notation}

Throughout we assume that $M^n$ is a compact, connected, oriented $n$-manifold with $n\ge 3$.
The set of smooth functions on $M$ is $C^\infty(M)$, and the positive
smooth functions are denoted by $C^\infty_+(M)$.
Given a bundle $E$ over $M$, we write $C^\infty(M,E)$ for smooth sections
of the bundle.  The bundle of symmetric $(0,2)$-tensors is $S_2 M$,
and $\calM$ is the set of smooth metrics on $M$ (i.e., the open
set of positive-definite elements of $C^\infty(M,S_2 M)$).  The
bundle of conformal classes of smooth metrics is $\calC$. All
objects in this paper are smooth.

We use a modified form of abstract index notation. 
Indices are used for tensorial objects
to clarify the number and
type of arguments, to help with contraction operations, and so forth, 
but are not associated with
the components of the tensor with respect to some specific coordinate
system. Whenever indices might clutter notation (e.g. when
the tensor is used as a subscript) we freely drop the indices.
So for a metric $g_{ab}$, the name of the metric is $g$ and
the indices are a helpful decoration to be used when they do not
get in the way.

The Levi-Civita connection of a metric $g_{ab}$ is $\nabla$ or $\tensor[^{g}]{\nabla}{}$ as needed,
and its (positively-oriented) volume form is $dV_g$.
Given a metric $g_{ab}$ and a function $\phi\in C^\infty_+(M)$,
we can form a conformally related metric 
\begin{equation}\label{eq:confchange}
\tilde g_{ab} = \phi^{q-2} g_{ab}
\end{equation} where $q$ is the 
dimensional constant
\begin{equation}\label{eq:q}
q = \frac{2n}{n-2}.
\end{equation}
All conformal transformations in this paper will have the form
\eqref{eq:confchange} since the scalar curvature $R_{\tilde g}$
of $\tilde g$ then has the simple form
\begin{equation}\label{eq:scalchange}
R_{\tilde g} = \phi^{1-q} (-2\kappa q \Lap_g \phi + R_g \phi),
\end{equation}
where $\Lap_g$ is the Laplacian of $g$ and
$\kappa$ is the dimensional constant
\begin{equation}\label{eq:kappa}
\kappa = \frac{n-1}{n}.
\end{equation}
Although the appearance of $\kappa$ in equation \eqref{eq:scalchange}
is somewhat awkward, it appears naturally throughout the 
equations connected to the conformal method, so we introduce this notation now.

The conformal class of $g_{ab}$ is $[g_{ab}]$.  When
we do not want to emphasize some particular representative
of the conformal class we use bold face instead: $\mathbf g$
denotes a conformal class as well. A boldface $\mathbf{g}$ 
and a plain $g$ are unrelated names, so an equation such as
$[g_{ab}]=\mathbf{g}$ is a nontrivial statement. Tangent and cotangent 
vectors to $\calC$ will be written with boldface as well.

Starting from a metric $g_{ab}$ and other conformal data
the conformal methods attempt to find a solution of the constraints
with a metric conformally related to $g_{ab}$.  We use overbars to 
denote conformally transforming objects that
satisfy the constraint equations, so $\ovl g_{ab}$ is the physical
solution metric.

\section{Conformal Objects}\label{sec:conformal}

Our goal here is to express the objects that appear in the various
conformal methods intrinsically with respect to a conformal class
rather than with respect to a representative metric. The set $\calC$ of smooth 
conformal classes can be shown to be a Fr\'echet manifold,
which provides a natural definition of tangent and cotangent vectors
at some conformal class $\mathbf{g}$.
To avoid this machinery, however, we take a more prosaic approach 
and define tangent and (certain) cotangent vectors to $\calC$ 
at $\mathbf g$
as tensorial objects that transform in a certain way when changing from 
one representative of $\mathbf{g}$ to another.  This is analogous to 
the old-fashioned approach of defining a manifold's tangent and cotangent vectors
as objects that transform in a certain way under coordinate changes.
We also give a related description of tangent and 
cotangent vectors to $\calC/\calD_0$ where $\calD_0$ is the
connected component of the identity of the diffeomorphism group.

\subsection{Conformal Tangent Vectors}\label{subsec:tangent}
Let $g^{0}_{ab}$ be a metric and let $g_{ab}(t)$ be a smooth path
with $g_{ab}(0)=g^0_{ab}$.  It is easy to see that if $g_{ab}(t)$ remains
in the conformal class $[g^0_{ab}]$ then there is function
$\alpha\in C^\infty(M)$ such that
\begin{equation}
g'_{ab}(0) = \alpha g^0_{ab}.
\end{equation}
Moreover, every smooth function arises this way for
some path (e.g. $g_{ab}(t) = e^{\alpha t}g^0_{ab}$).  So we identify
\begin{equation}
\left\{\alpha g_{ab}^0: \alpha\in C^{\infty}(M)\right\}
\end{equation}
as the tangent space of the conformal class $\left[g_{ab}^0\right]$ at $g^0_{ab}$.
Given an arbitrary path starting at $g^0_{ab}$ we can 
uniquely decompose
\begin{equation}
g'_{ab}(0) = u_{ab} + \alpha g^0_{ab}
\end{equation}
where $u_{ab}$ is trace-free with respect to $g^0_{ab}$ and
$\alpha\in C^\infty(M)$. It is therefore natural
to identify the trace-free tensors $u_{ab}$ as the directions of travel through the set of conformal classes.
Given a smooth function $\beta(t)$, the paths $g_{ab}(t)$ and 
\begin{equation}
\tilde g_{ab}(t) = e^{\beta(t)} g_{ab}(t)
\end{equation}
descend to the same path in $\calC$. Since
$
\tilde g'_{ab}(0) = e^{\beta} u_{ab}  + (\beta' e^{\beta} +\alpha) g_{ab}^0
$
we therefore identify 
$u_{ab}$ at $g_{ab}$ 
and
$e^{\beta} u_{ab}$ at $e^\beta g_{ab}$
as representing the same tangent vector
to $\calC$ at $\left[g_{ab}^0\right]$.
\begin{definition}\label{def:tangent}
Let $\calX$ be the set of pairs $(g_{ab}, u_{ab})$
where $g_{ab}\in \calM$ and where $u_{ab}\in C^\infty(M,S_2 M)$
is trace-free with respect to $g_{ab}$.
A \textbf{conformal tangent vector} is an element of
$
\calX / \sim
$
where 
\begin{equation}
(\tilde g_{ab},\tilde u_{ab}) \sim (g_{ab},u_{ab})
\end{equation}
if there exists $\phi\in C_+^\infty(M)$ such that
\begin{equation}
\ifjournal\eqalign{
\tilde g_{ab} &= \phi^{q-2} g_{ab}\\
\tilde u_{ab} &= \phi^{q-2} u_{ab}.}
\else
\begin{aligned}
\tilde g_{ab} &= \phi^{q-2} g_{ab}\\
\tilde u_{ab} &= \phi^{q-2} u_{ab}.
\end{aligned}
\fi
\end{equation}
We use the following notation:
\begin{itemize}
\item $[ g_{ab},u_{ab}]$ is the conformal tangent vector
corresponding to $(g_{ab},u_{ab})$.
\item For each $\mathbf{g}\in \calC$, $T_\mathbf{g} \calC$ 
is the set of conformal tangent vectors $[g_{ab},u_{ab}]$ with
$g_{ab}\in \mathbf{g}$.
\item $T \calC = \cup_{\mathbf{g}\in \calC} T_{\mathbf{g}}\calC$.
\end{itemize}
More generally, if $g_{ab}$ is a metric and $S_{ab}$ is an arbitrary 
symmetric $(0,2)$-tensor field, we define
\begin{equation}
[g_{ab}, S_{ab}] = [g_{ab}, u_{ab}]
\end{equation}
where $u_{ab}$ is the trace-free part of $S_{ab}$ (as computed with respect to $g_{ab}$).
This should be thought of as the pushforward of the tangent vector $S_{ab}$ to the
space of metrics at $g_{ab}$ to an element of $T_{[g]}\calC$ under the natural projection.
\end{definition}

Suppose $\mathbf{g}$ is a conformal class and $\mathbf{u}\in T_{\mathbf{g}} \calC$.
Given a representative $g_{ab}\in \mathbf{g}$ it is clear that there is a unique
trace-free $u_{ab}\in C^\infty(M,S_2 M)$ with $[g_{ab}, u_{ab}]=\mathbf{u}$, which we will call
the \textbf{representative of $\mathbf u$ with respect to $g_{ab}$}.  We give
$T_{\mathbf{g}} \calC$ the topology of the subspace of $C^\infty(M,S_2 M)$ 
determined by this identification and note that the topology
is independent of the choice of representative $g_{ab}$.

\subsection{The Conformal Killing Operator}
The conformal tangent vectors that arise by flowing a
conformal class $\mathbf{g}$ through a path of diffeomorphisms can be described in terms of a map
$L_{\mathbf g}:C^\infty (M,TM)\ra T_{\mathbf g}\calC$ called the conformal Killing operator.

Let $\mathbf g$ be a conformal class with representative $g_{ab}$ and 
suppose $\Phi_t$ is a path of diffeomorphisms starting at the identity with infinitesimal 
generator $X^a$.
The path
\begin{equation}
h_{ab}(t) = \Phi_t^* g_{ab},
\end{equation}
satisfies
\begin{equation}
h_{ab}'(0) = \calL_{X} g_{ab},
\end{equation}
where $\calL_X g_{ab}$ is the Lie derivative of $g_{ab}$ with respect
to the vector field $X^a$.  We decompose $\calL_X g_{ab}$ into its trace and trace-free
parts with respect to $g_{ab}$ to obtain
\begin{equation}
\calL_X g_{ab} = (\ck_g X)_{ab} + \frac{2\div_g X}{n} g_{ab}
\end{equation}
where $\div_g X = \tensor[^g]{\nabla}{_a} X^a$ and where
\begin{equation}
(\ck_g X)_{ab} = \tensor[^g]{\nabla}{_a} X_b + \tensor[^g]{\nabla}{_b} X_a -\frac{2\div_g X}{n} g_{ab}
\end{equation}
is the usual conformal Killing operator. The conformal tangent vector
\begin{equation}
\mathbf{u} = [g_{ab}, (\ck_{g} X)_{ab}].
\end{equation}
does not depend on the choice of representative of $\mathbf{g}$.
Indeed, an easy computation shows that if 
$\tilde g_{ab}$ is another metric conformally related to
$g_{ab}$ via $\tilde g_{ab} = \phi^{q-2} g_{ab}$ then 
\begin{equation}
\ck_{\tilde g} X =  \phi^{q-2} \ck_{g} X
\end{equation}
and consequently
\begin{equation}
[\tilde g_{ab}, (\ck_{\tilde g} X)_{ab}]=[g_{ab}, (\ck_g X)_{ab}] = \mathbf{u}.
\end{equation}
We therefore obtain a well-defined conformal Killing operator
$\ck_{\mathbf{g}}: C^\infty(M,TM) \ra T_{\mathbf g}\calC$ 
given by
\begin{equation}
\ck_{\mathbf{g}} X =  [g_{ab},\ck_{g}X].
\end{equation}
for any representative $g_{ab}$ of $\mathbf{g}$.  Since the map 
$\ck_g :  C^\infty(M,TM) \ra C^\infty(M,S_2 M)$ is continuous, and 
since the projection $u_{ab} \ra [g_{ab}, u_{ab}]$ is continuous,
so is $\ck_{\mathbf{g}}$. 

The elements of the kernel of $\ck_{\mathbf{g}}$
are called \textbf{conformal Killing fields}.  Generically there are none \cite{Fischer:1977wm}.

\subsection{Conformal Cotangent Vectors}\label{subsec:cotangent}
The conformal method involves symmetric, trace-free, (0,2)-tensors $S_{ab}$ that obey
the conformal transformation law
\begin{equation}
\tilde A_{ab} = \phi^{-2} A_{ab}
\end{equation}
when $g_{ab}$ is transformed to $\tilde g_{ab} = \phi^{q-2}g_{ab}$.
Such objects are not tangent vectors to $\calC$
since the wrong power of $\phi$
appears in the transformation law; rather, these encode 
cotangent vectors as follows.

Given the pair $(g_{ab},A_{ab})$ we define a functional
on symmetric trace free tensors $u_{ab}$ via
\begin{equation}
F_{g,A}(u) = \int_M \ip<A,u>_g dV_g.
\end{equation}
If we conformally transform $u_{ab}$ as a tangent vector
\begin{equation}
\tilde u_{ab} = \phi^{q-2} u_{ab}
\end{equation}
and we transform 
\begin{equation}\label{eq:cotangent}
\tilde A_{ab} = \phi^{-2} A_{ab}
\end{equation}
then
\begin{equation}
\ip< \tilde A, \tilde u>_{\tilde g} = \phi^{2-q}\phi^{2-q}\ip<\phi^{-2} A, \phi^{q-2}u>_{g} = \phi^{-q}\ip<A,u>_g.
\end{equation}
At the same time, the volume form transforms as $dV_{\tilde g} = \phi^q dV_g$ and therefore
\ifjournal
\begin{eqnarray}
F_{\tilde g, \tilde A}(\tilde u) =
\int_{M} \ip<\tilde A,\tilde u>_{\tilde g} dV_{\tilde g} =\\
\qquad=  \int_{M} \phi^{-q}\ip<A,u>_g \phi^q dV_g = \int_M \ip<A,u>_g dV_g = F_{g,A}(u).
\end{eqnarray}
\else
\begin{equation}
F_{\tilde g, \tilde A}(\tilde u) =
\int_{M} \ip<\tilde A,\tilde u>_{\tilde g} dV_{\tilde g} = \int_{M} \phi^{-q}\ip<A,u>_g \phi^q dV_g
= \int_M \ip<A,u>_g dV_g = F_{g,A}(u).
\end{equation}
\fi
Thus we can associate with $(g_{ab},A_{ab})$ a well-defined functional on $T_{[g]}\calC$
when we transform $A_{ab}$ according to equation \eqref{eq:cotangent}.
\begin{definition}\label{def:cotangent}
Let $\calY$ be the set of pairs $(g_{ab}, A_{ab})$ where $g_{ab}\in\calM$ and
where $A_{ab}\in C^\infty(M,S_2 M)$ is trace-free with respect to $g_{ab}$.
A \textbf{(smooth) conformal cotangent vector} is an 
equivalence class of $\calY/\sim$ under the relation
\begin{equation}
(\tilde g_{ab},\tilde A_{ab}) \sim (g_{ab},A_{ab})
\end{equation}
if there is a smooth positive function $\phi$ on $M$ such that
\begin{equation}
\ifjournal\eqalign{
\tilde g_{ab} &= \phi^{q-2} g_{ab}\\
\tilde A_{ab} &= \phi^{-2} A_{ab}.	
}\else
\begin{aligned}
\tilde g_{ab} &= \phi^{q-2} g_{ab}\\
\tilde A_{ab} &= \phi^{-2} A_{ab}.
\end{aligned}\fi
\end{equation}
We use the following notation:
\begin{itemize}
\item $[ g_{ab},A_{ab}]^*$ is the conformal cotangent vector
corresponding to $(g_{ab},A_{ab})$.
\item For each $\mathbf{g}\in \calC$, $T^*_\mathbf{g} \calC$ 
is the set of conformal tangent vectors $[g_{ab},A_{ab}]^*$ with
$g_{ab}\in \mathbf{g}$.
\item $T^*\calC = \cup_{\mathbf{g}\in \calC} T^*_{\mathbf{g}}\calC$.
\end{itemize}

More generally, if $K_{ab}$ is an arbitrary element of $C^\infty(M, S_2 M)$ we define
\begin{equation}
[g_{ab}, K_{ab}]^* = [g_{ab}, A_{ab}]^*
\end{equation}
where $A_{ab}$ is the trace-free part of $K_{ab}$ (with respect to $g_{ab}$).

Given a smooth conformal tangent vector $\mathbf{A}\in T^*_{\mathbf{g}}\calC$ and a conformal tangent vector
$\mathbf{u}\in T_{\mathbf{g}}\calC$, we define
\begin{equation}
\ip<\mathbf{A},\mathbf{u}> = \int \ip<A,u>_g\;dV_g
\end{equation}
where $(g_{ab}, A_{ab})$ and $(g_{ab},u_{ab})$ are any representatives
of $\mathbf{A}$ and $\mathbf{u}$ with respect to the same background metric $g_{ab}\in\mathbf{g}$.
This linear map is evidently continuous, so we identify $T_{\mathbf{g}}^* \calC$ 
with a subspace of $(T_{\mathbf{g}}\calC)^*$.

\end{definition}
The containment $T^*_{\mathbf{g}}\calC \subseteq (T_{\mathbf{g}}\calC)^*$
is strict since the topological dual space contains more general distributions,
which motivates the modifier \textit{smooth} in the previous definition.
Given that we represent conformal tangent vectors using symmetric trace-free
$(0,2)$-tensor fields, it may be more natural to represent conformal cotangent vectors
using symmetric trace-free $(2,0)$-tensor fields.  To this end, we also define
\begin{equation}
[g_{ab}, A^{ab}]^* = [g_{ab}, g_{ac}g_{bd} A^{cd}]^*.
\end{equation}
It is easy to see that if $\tilde g_{ab} = \phi^{q-2} g_{ab}$, then
\begin{equation}
[\tilde g_{ab}, \tilde A^{ab}]^* = [g_{ab}, A^{ab}]^*
\end{equation}
if and only if $\tilde A^{ab} = \phi^{2-2q}A^{ab}$, which recovers another 
familiar transformation law for the conformal method.  Symmetric trace-free
tensors transforming according to $\tilde A_{ab}=\phi^{-2} A_{ab}$
or $\tilde A^{ab} = \phi^{2-2q} A^{ab}$ are both representations of
conformal cotangent vectors.

The distinction between tangent and cotangent vectors is important
because unlike the situation for the space $\calM$ of metrics,
we do not have a canonical identification of tangent and cotangent
vectors for $\calC$.
The tangent space of $\calM$ at a metric $g_{ab}$ is 
$T_g\calM = C^\infty(M,S_2 M)$ and is equipped with a
natural metric defined by
\begin{equation}
\ip<h,k> = \int_M \ip<h,k>_g dV_g.
\end{equation}
The metric provides a natural identification of 
$T_g\calM$ with a subspace of $(T_g\calM)^*$
by taking to $h_{ab}$ to $\ip<h,\cdot>$.
Unfortunately, this inner product does not descend
to an inner product on $T_{[g]}\calC$, and we
do not have a canonical way to identify $T_{[g]} \calC$ 
with a subspace of $(T_{[g]}\calC)^*$.
Instead,
we have a family of identifications depending on the choice of
a volume form on $M$.

\begin{proposition}\label{prop:TvsTstar}
Let $\omega$ be a smooth volume form on $M$ (i.e., a nonvanishing, 
positively-oriented section of $\Lambda^n M$).  There is a unique linear map
$k_{\omega}: T\calC \ra T^*\calC$ satisfying the following:
\begin{itemize}
\item For each $\mathbf{g}\in \calC$, $k_{\omega}:T_{\mathbf{g}}\calC \ra T_{\mathbf{g}}^*\calC$
is continuous and bijective.
\item If $\mathbf{u},\mathbf{v}\in T_{\mathbf{g}} \calC$, and if 
$u_{ab}$ and $v_{ab}$ are their representatives with respect to some
common background metric $g_{ab}\in\mathbf g$, then
\begin{equation}\label{eq:TTstartindent}
\ip<k_{\omega}(\mathbf{u}),\mathbf{v}> = \int_M \ip<u,v>_g \,\omega
\end{equation}
\item If $\mathbf{g}$ is a conformal class
and $\mathbf{u} \in T_{\mathbf g} \calC$
with representative $(g_{ab}, u_{ab})$ then
\begin{equation}\label{eq:ttstarlapse}
k_\omega(\mathbf{u}) = [g_{ab}, (\omega/dV_g)\, u_{ab}]^*.
\end{equation}

\end{itemize}
\end{proposition}
\begin{proof}
We define $k_\omega$ by equation \eqref{eq:TTstartindent} and need to show that
it is well-defined and has the stated properties.  To see that it is 
well-defined, suppose $g_{ab}$ and $\tilde g_{ab} = \phi^{q-2} g_{ab}$ are two
representatives of a conformal class $\mathbf{g}$ and suppose 
$\mathbf{u},\mathbf{v}\in T_{\mathbf{g}} \calC$.  Let $u_{ab}$, $\tilde u_{ab}$,
$v_{ab}$, and $\tilde v_{ab}$ be the representatives of $\mathbf{u}$ and $\mathbf{v}$
with respect to $g_{ab}$ and $\tilde g_{ab}$, so $\tilde u_{ab} = \phi^{q-2} u_{ab}$
and similarly with $\tilde v_{ab}$.  Then
\begin{equation}
\ip<\tilde u, \tilde v>_{\tilde g} = \phi^{4-2q}\ip< \tilde u, \tilde v>_g = 
\phi^{4-2q}\ip< \phi^{q-2} u, \phi^{q-2} v>_g = \ip<u,v>_g.
\end{equation}
Thus
\begin{equation}
\int_{M} \ip<\tilde u,\tilde v>_{\tilde g} \,\omega = \int_{M} \ip<u,v>_g \,\omega 
\end{equation}
and $k_\omega$ is well-defined.

It is clear that $k_\omega(\mathbf{u}) \in (T_{\mathbf g} \calC)^*$.  To see
that it belongs to $T^*_{\mathbf g} \calC$, pick a representative $g_{ab}\in \mathbf{g}$,
let 
\begin{equation}\label{eq:uToS}
\mathbf{A} = [g_{ab}, (\omega/dV_g)\, u_{ab}]^*.
\end{equation}
Then for any $\mathbf v\in T_{\mathbf{g}}\calC$
\begin{equation}
\ip<\mathbf A,\mathbf v> = \int_M \ip<(\omega/dV_g)\,u,v>_g dV_g = \int_M \ip<u,v>_g \omega = \ip<k_\omega(\mathbf{u}),\mathbf{v}>
\end{equation}
and hence $k_\omega(\mathbf{u}) = \mathbf{A}\in T^*_\mathbf{g}\calC$.

To see that $k_\omega$ is bijective as a map into $T^*_\mathbf{g}\calC$ we first note that
$\ip<k_\omega(\mathbf{u}),\mathbf{u}>>0$ unless $\mathbf{u}=0$ and $k_\omega$ is 
therefore injective.
Consider an arbitrary $\mathbf{A}\in T^*_\mathbf{g}\calC$ and
write $\mathbf{A} = [g_{ab}, A_{ab}]^*$.  The previous 
computation shows that $k_{\omega}([g_{ab}, (dV_g/\omega) A_{ab}]) = 
\mathbf{A}$ so $k_\omega$ is surjective as well.  The continuity of $k_\omega$
is a straightforward consequence of the fact that the right-hand side of
\eqref{eq:TTstartindent} defines a continuous map $C^0(M,S_2M)\ra(C^0(M,S_2M))^*$
and the continuity of the embeddings $C^\infty(M,S_2 M)\hookrightarrow C^0(M,S_2M)$ and
$(C^0(M,S_2M))^* \hookrightarrow (C^\infty(M,S_2 M))^*$.
\end{proof}

A metric $g_{ab}$ determines a volume form $dV_{g}$ and there is a one-to-one correspondence
between metrics and pairs $(\mathbf{g},\omega)$ of conformal classes and volume forms.  So
the choice of a volume form $\omega$ in Proposition \ref{prop:TvsTstar} can be thought of,
at least when working with $T_{\mathbf{g}}\calC$ for some fixed $\mathbf{g}$,
as a choice of representative metric within the conformal class.  

\subsection{The Divergence}
We have previously defined the conformal Killing operator associated with a 
conformal class $\mathbf{g}$, 
\begin{equation}
\ck_{\mathbf{g}} : C^\infty (M,TM) \ra T_{\mathbf{g}}\calC.
\end{equation}
This is a continuous linear map, and hence we obtain a continuous adjoint
\begin{equation}
\ck_{\mathbf g}^* : (T_{\mathbf g}\calC)^* \ra  (C^\infty(M,TM))^*
\end{equation}
given by
\begin{equation}
\ip< \ck_{\mathbf{g}}^* (F), X> = \ip<F, \ck_{\mathbf{g}} X>.
\end{equation}
We define the divergence
\begin{equation}
\div_{\mathbf g} = -\frac{1}{2} \ck_{\mathbf g}^*
\end{equation}

Note that if $\mathbf{A} = [g_{ab}, A_{ab}]^*$ is a smooth cotangent vector and $X^a$ is a smooth
vector field then, using 
the definition of $\ck_{[g]}$ and
integration by parts, we find
\begin{equation}\label{eq:div}
\ifjournal\eqalign{
\ip< \div_{[g]} \mathbf{A}, X>  &= -\frac{1}{2}\ip< \mathbf{A}, \ck_{[g]} X>\\
&= -\frac{1}{2} \ip< [g_{ab},A_{ab}]^*, [g_{ab}, \ck_g X]>\\
&= -\frac{1}{2} \int \ip<A,\ck_g X>_g dV_g\\
&= \int (\div_g A)_a X^a dV_g.	
}\else
\begin{aligned}
\ip< \div_{[g]} \mathbf{A}, X>  &= -\frac{1}{2}\ip< \mathbf{A}, \ck_{[g]} X>\\
&= -\frac{1}{2} \ip< [g_{ab},A_{ab}]^*, [g_{ab}, \ck_g X]>\\
&= -\frac{1}{2} \int \ip<A,\ck_g X>_g dV_g\\
&= \int (\div_g A)_a X^a dV_g.
\end{aligned}\fi
\end{equation}

\subsection{Quotients Modulo Flows}\label{subsec:tt}

The space of \textbf{conformal geometries}, sometimes called \textbf{conformal superspace},
is the quotient of $\calC$ obtained by identifying conformal classes
if there is a flow taking one to another.  We will write this quotient
symbolically as $\calC/\calD_0$ (here $\calD_0$ is the connected component of the identity
of the diffeomorphism group). Because the Einstein equations are diffeomorphism
invariant, the space $\calC/\calD_0$ is more fundamental than $\calC$, and it will
be important to work with tangent and cotangent vectors to this space.

Suppose we have a curve $\gamma$ 
of conformal classes obtained by a flow. Its tangent vector at $\mathbf{g}=\gamma(0)$
is then $\ck_{\mathbf g} X$ for some vector field $X^a$.
Since $\gamma$ descends to a stationary curve in $\calC/\calD_0$ the directions $\Im \ck_{\mathbf{g}}$
become null directions in $\calC/\calD_0$,  which motivates the following.
\begin{definition}\label{def:cgv} Let $\mathbf g\in \calC$.  The space
of \textbf{conformal geometric velocities} at the conformal geometry represented by
$\mathbf g$ is the quotient space 
\begin{equation}
T_{\mathbf g}\calC/ \Im \ck_{\mathbf g}.
\end{equation}
The conformal geometric velocity represented by a conformal tangent vector $\mathbf u$
is the subspace
\begin{equation}
[\mathbf u] = \mathbf u + \Im \ck_{\mathbf{g}}
\end{equation}
of $T_{\mathbf g} \calC$.
We write $T_{\mathbf g} (\calC/\calD_0)$ for the set of conformal velocities at the conformal
geometry represented by $\mathbf g$. 
\end{definition}

Note we are deliberately avoiding working with equivalence classes $[\mathbf g]$ of conformal
classes under flows, and that each representative of $[\mathbf g]$ gives a
representation $T_{\mathbf g} (\calC/\calD_0)$
of an object that would be written as $T_{[\mathbf g]}(\calC/\calD_0)$.

Every conformal tangent vector $\mathbf u \in T_{\mathbf g} \calC$ naturally determines
the conformal geometric velocity $\mathbf u + \Im \ck_{\mathbf g}$.  
Fixing a representative metric $g_{ab}$ of $\mathbf g$, the conformal geometric velocities
at $\mathbf g$ are naturally identified with the subspaces
\begin{equation}
u_{ab} + \Im \ck_{g}
\end{equation}
of $C^\infty(M,S_2 M)$
where $u_{ab}$ is trace-free with respect to $g_{ab}$.

Elements of the dual space $(T_{\mathbf g}\calC/ \Im \ck_{\mathbf g})^*$ can be represented as
elements of the subspace of $(T_{\mathbf g}\calC)^*$ that annihilate
 $\Im \ck_{\mathbf g}$.  Restricting our attention
to those elements that are also \textit{smooth} conformal cotangent vectors 
we have the following.
\begin{definition} Let $\mathbf g\in \calC$.  The space
of \textbf{conformal geometric momenta} at the conformal geometry represented by
$\mathbf g$ is  
the subspace of $T^*_{\mathbf g}\calC$ consisting of those elements
that vanish on $\Im \ck_{\mathbf g}$. We denote this subspace by
$T^*_\mathbf{g} (\calC/\calD_0)$.
\end{definition}

The subspace $T^*_\mathbf{g} (\calC/\calD_0)\subseteq T^*_\mathbf  g\calC$ of conformal
geometric momenta can be characterized in terms of the divergence $\div_\mathbf g$,
and this leads to the notion of a 
transverse traceless tensor.
A symmetric tensor $\sigma_{ab}$ is said to be 
\textbf{transverse traceless} (TT) with respect to a metric $g_{ab}$ if
it is traceless,  
\begin{equation}
g^{ab}\sigma_{ab}=0,
\end{equation}
and transverse,
\begin{equation}
\div_{g} \sigma = 0.
\end{equation}
Lichnerowicz observed \cite{Lichnerowicz:1944} that TT tensors
behave well with respect to conformal transformations:
if $\sigma_{ab}$ is TT with respect to
$g_{ab}$, then $\tilde\sigma_{ab} = \phi^{-2}\sigma_{ab}$
is TT with respect to $\tilde g_{ab} = \phi^{2-q} g_{ab}$.
From this conformal transformation law we identify
\begin{equation}
\bfsigma = [g_{ab}, \sigma_{ab}]^*
\end{equation}
as a smooth conformal cotangent vector.  Moreover, equation \eqref{eq:div}
implies that $\div_{[g]} \bfsigma = 0$.  The following easy lemma
shows that the TT tensors
represent the smooth conformal cotangent vectors that annihilate the image of the
conformal Killing operator; we omit the proof.

\begin{lemma}\label{lem:whatsaTT} For $\bfsigma\in T^*_{\mathbf{g}}\calC$ the following are equivalent.
\begin{enumerate}
\item $\div_{\mathbf{g}} \bfsigma = 0$.
\item For all smooth vector fields $X^a$
\begin{equation}
\ip< \bfsigma, \ck_{\mathbf{g}} X> = 0.
\end{equation}
\item For some $g_{ab}$ and $\sigma_{ab}$ with 
$\bfsigma=[g_{ab},\sigma_{ab}]^*$, $\sigma_{ab}$
is TT with respect to $g_{ab}$.
\item For all $g_{ab}$ and $\sigma_{ab}$ with 
$\bfsigma=[g_{ab},\sigma_{ab}]^*$, $\sigma_{ab}$
is TT with respect to $g_{ab}$.
\end{enumerate}
\end{lemma}

As a consequence of Lemma \ref{lem:whatsaTT} we have shown
\begin{equation}
T^*_{\mathbf g}(\calC/\calD_0) = \{\bfsigma\in T^*\calC: \div_{\mathbf g} \sigma = 0\}.
\end{equation}
So transverse traceless tensors are the representations, in terms of a background metric,
of conformal geometric momenta.

We have seen that a conformal tangent vector $\mathbf u$ naturally defines
a conformal geometric velocity $\mathbf u + \Im \ck_{\mathbf g} \in T_\mathbf g(\calC/\calD_0)$.
On the other hand, an arbitrary conformal cotangent vector $\mathbf{A}$ does \textbf{not} naturally determine
a conformal geometric momentum: this would require a choice of projection from
$T^*_\mathbf g\calC$  onto the subspace $T^*_\mathbf g(\calC/\calD_0)$. Our next goal is
to describe a family of such projections that are closely related to the maps $k_\omega$
from Proposition \ref{prop:TvsTstar}.  To begin, we recall the following result 
from \cite{York:1973fl}, which is a fundamental component of the 1974 conformal method.
\begin{proposition}[York Splitting]\label{prop:york-split} Let $g_{ab}$ be a smooth
Riemannian metric on $M$ and let $A_{ab}$ 
be a smooth, trace-free, symmetric (0,2)-tensor field.
Then there is a smooth TT tensor field
$\sigma_{ab}$ and a smooth vector field $X^a$ such
that
\begin{equation}\label{eq:yorksplit}
A_{ab} = \sigma_{ab} + (\ck_{g} X)_{ab}.
\end{equation}
This decomposition is unique up to the addition of a conformal Killing field to $X^a$.
\end{proposition}
Notice that the right-hand side of equation \eqref{eq:yorksplit} does not have a
natural interpretation as a conformal object: it is the sum of a representative 
$\sigma_{ab}$ of a conformal cotangent vector
with a representative $(\ck_g X)_{ab}$ of a conformal tangent vector.
Adding these together requires an identification of $T_{\mathbf{g}}\calC$ with  
$T^*_{\mathbf{g}}\calC$.  We can reformulate Proposition \ref{prop:york-split}, however,
in terms of conformal objects 
using the maps $k_\omega$ defined in Proposition \ref{prop:TvsTstar}.

\begin{proposition}\label{prop:yorksplitconf}
Let $\mathbf{g}\in\calC$ and let $\mathbf{A}\in T^*_\mathbf{g} \calC$.  Given a choice of a volume form $\omega$,
there is a conformal geometric momentum $\bfsigma$ and a vector field $X^a$ such that
\begin{equation}\label{eq:yorksplitconf}
\mathbf{A} = \bfsigma + k_\omega( \ck_\mathbf{g} X),
\end{equation}
where $k_\omega$ is the map defined in Proposition \ref{prop:TvsTstar}. 
The decomposition is unique up the the addition of a conformal Killing field to $X^a$.

Moreover, if $g_{ab}\in \mathbf{g}$ is the unique metric with $dV_g = \omega$, and
if $A_{ab}$ and $\sigma_{ab}$ are the representatives of 
$\mathbf{A}$ and $\bfsigma$ with respect to $g_{ab}$, then
\begin{equation}\label{eq:Asplitconf}
A_{ab} = \sigma_{ab} + (\ck_{g} X)_{ab}.
\end{equation}
\end{proposition}
\begin{proof}
Let $g_{ab}$ be the unique metric in $\mathbf{g}$ with $dV_g = \omega$.  
We wish to write equation \eqref{eq:yorksplitconf} in terms of representatives
with respect to $g_{ab}$.  

Proposition \ref{prop:TvsTstar} equation \eqref{eq:ttstarlapse}
implies that for any $[g_{ab},v_{ab}]\in T_{\mathbf{g}} \calC$,
\begin{equation}
k_\omega(\mathbf{v}) = [g_{ab}, v_{ab}]^*;
\end{equation}
this is the step where we use the specific choice of $g_{ab}$.
In particular, for a vector field $X^a$,
\begin{equation}
k_{\omega}(\ck_{\mathbf{g}} X) = [g_{ab}, (\ck_g X)_{ab}]^*.
\end{equation}
Hence equation \eqref{eq:yorksplitconf} is equivalent to finding
a TT tensor $\sigma_{ab}$ and vector field $X^a$ such that
\begin{equation}
[g_{ab},A_{ab}]^* = [g_{ab},\sigma_{ab}]^* + [g_{ab},(\ck_{g} X)_{ab}]^*,
\end{equation}
where $A_{ab}$ is the tensor field such that $\mathbf{A} = [g_{ab}, A_{ab}]^*$.
In other words, we wish to solve
\begin{equation}
A_{ab} = \sigma_{ab} + (\ck_{g} X)_{ab}.
\end{equation}
and the result now follows from Proposition \ref{prop:york-split}.
\end{proof}

For each choice of volume form $\omega$, Proposition
\ref{prop:yorksplitconf} determines a projection from the space 
of conformal cotangent vectors onto the subspace of conformal
geometric momenta.
\begin{definition}\label{def:proj}
Let $\omega$ be a volume form. For all $\mathbf g\in \calC$, the projection
 $P_\omega: T^*_{\mathbf g} \calC \ra T^*_{\mathbf g} (\calC/\calD_0)$
is defined by
\begin{equation}
P_\omega(\mathbf{A}) = \bfsigma
\end{equation}
where $\bfsigma$ is the unique conformal geometric momentum determined by equation \eqref{eq:yorksplitconf}.
\end{definition}

The maps $k_\omega$ each determine identifications of $T_\mathbf{g}\calC$ with 
$T^*_\mathbf{g}\calC$. Using the projections $P_\omega$ we can now define related
identifications $j_\omega:T_\mathbf{g} (\calC/\calD_0)\ra T_\mathbf{g}^* (\calC/\calD_0)$
that satisfy
\ifjournal
\begin{equation}\label{eq:ij}
\xymatrix{T_{\mathbf g} \calC \ar[d]_{\pi} & \ar[l]_{k_\omega^{-1}} T_{\mathbf g}^*\calC  \\
T_\mathbf g (\calC/\calD_0) & \ar[l]^{j_\omega^{-1}} T^*_\mathbf g (\calC/\calD_0)\ar[u]_{\pi^*}, }
\end{equation}
\else
\begin{equation}\label{eq:ij}
\begin{gathered}
\xymatrix{T_{\mathbf g} \calC \ar[d]_{\pi} & \ar[l]_{k_\omega^{-1}} T_{\mathbf g}^*\calC  \\
T_\mathbf g (\calC/\calD_0) & \ar[l]^{j_\omega^{-1}} T^*_\mathbf g (\calC/\calD_0)\ar[u]_{\pi^*}, }
\end{gathered}
\end{equation}\fi
where $\pi$ is the natural projection and $\pi^*$ is the natural embedding.
\begin{definition}\label{def:j}
Let $\omega$ be a volume form. For each $\mathbf g\in\calC$ we define 
$j_\omega: T_{\mathbf g} (\calC/\calD_0)\ra  T_{\mathbf g}^* (\calC/\calD_0)$
by
\begin{equation}\label{eq:jdef}
j_\omega( \mathbf{u} + \Im \ck_\mathbf g) = P_\omega(k_\omega(\mathbf u)).
\end{equation}
Its inverse is given by
\begin{equation}\label{eq:jinvdef}
j^{-1}_\omega(\bfsigma) = k^{-1}_\omega (\bfsigma) + \Im \ck_{\mathbf g}.
\end{equation}
\end{definition}
One needs to verify that $j_\omega$ is well-defined, but this is an easy consequence
of the uniqueness clause of Proposition \ref{prop:yorksplitconf}. Showing that $j_\omega^{-1}$ really
is the inverse of $j_\omega$ is also an easy exercise using Proposition \ref{prop:yorksplitconf}
and is left to the reader.  Note that the commutative diagram \eqref{eq:ij} is simply an alternative
expression of equation \eqref{eq:jinvdef}.

\section{The CMC Conformal Method}\label{sec:CMC}
Suppose $(\ovl g_{ab}, \ovl K_{ab})$ is a solution of the constraints such that 
$\ovl g^{ab}\ovl K_{ab}=\tau_0$ for some constant $\tau_0$; we say such
a solution is constant mean curvature or CMC. Letting $\ovl \sigma_{ab}$ be the 
trace-free part of $K_{ab}$ the momentum constraint \eqref{eq:momentum}
then reads
\begin{equation}
\div_{g} \ovl \sigma = 0
\end{equation}
and hence $\ovl\sigma_{ab}$ is TT.  So
\begin{equation}\label{eq:CMC-K-decomp}
\ovl K_{ab} = \ovl \sigma_{ab} + \frac{\tau_0}{n}\ovl g_{ab}
\end{equation}
for some unique TT tensor $\ovl \sigma_{ab}$ and constant $\tau_0$.
In this way, every 
CMC solution determines a unique conformal class $\mathbf{g}=[\ovl g_{ab}]$,
conformal geometric momentum 
$\bfsigma=[\ovl g_{ab}, \ovl \sigma_{ab}]^*\in T^*_{\mathbf{g}}(\calC/\calD_0)$,
and constant $\tau_0$. We refer to a triple
$(\mathbf{g}, \bfsigma, \tau_0)$ as \textbf{CMC conformal data}\footnote{
Since $\bfsigma$ determines $\mathbf{g}$ (every cotangent vector determines its
base point), this description of conformal data is mildly redundant. Nevertheless, it
is useful to have an explicit notation for the conformal class.}.
% and  a triple $(g_{ab},\sigma_{ab},\tau_0)$ as {\bf representative CMC conformal data}.

The CMC conformal method of \cite{York:1973fl}
seeks to reverse this process:
starting from 
CMC conformal data $(\mathbf{g},\bfsigma,\tau_0)$ we wish to construct a CMC solution
$(\ovl g_{ab}, \ovl K_{ab})$ of the constraints with 
\begin{equation}\label{eq:CMCconds}
\ifjournal\eqalign{
\strut[ \ovl g_{ab} ] &= \mathbf{g}\\
\strut[ \ovl g_{ab},  \ovl K_{ab}]^* &= \bfsigma\\
\ovl g^{ab}  \ovl K_{ab} &= \tau_0.	}
\else
\begin{aligned}
\strut[ \ovl g_{ab} ] &= \mathbf{g}\\
\strut[ \ovl g_{ab},  \ovl K_{ab}]^* &= \bfsigma\\
\ovl g^{ab}  \ovl K_{ab} &= \tau_0.
\end{aligned}\fi
\end{equation}
To solve this problem, let $g_{ab}$ be an arbitrary representative of $\mathbf{g}$
and let $\sigma_{ab}$ be the unique TT tensor such that 
\begin{equation}
[g_{ab}, \sigma_{ab}] = \bfsigma.
\end{equation}
If $(\ovl g_{ab},\ovl K_{ab})$ is a solution of the constraints satisfying \eqref{eq:CMCconds},
then there is a conformal factor $\phi$ such that $\ovl g_{ab} = \phi^{q-2}g_{ab}$
and such that $\ovl K_{ab}$ satisfies equation \eqref{eq:CMC-K-decomp} with 
$\ovl \sigma_{ab} = \phi^{-2}\sigma_{ab}$.  Writing the constraint
equations \eqref{eq:constraints} in terms of $g_{ab}$ and $\sigma_{ab}$,
we find that the momentum constraint is automatically satisfied and 
(using the scalar curvature transformation law \eqref{eq:scalchange})
the Hamiltonian constraint is equivalent to
\begin{equation}\label{eq:lich}
-2\kappa q\, \Lap_g \phi + R_g\phi -\left|\sigma\right|_{g}^2\phi^{-q-1} + \kappa \tau_0^2\phi^{q-1} = 0,
\ifjournal\quad\else\qquad\fi\text{\small[Lichnerowicz-York equation]}
\end{equation}
where $R_g$ is the scalar curvature of $g_{ab}$ and 
$\kappa$ and $q$ are the dimensional constants defined by 
equations \eqref{eq:q} and \eqref{eq:kappa}.  Thus we have established the following.

% In practice we specify conformal CMC data $(\mathbf{g},\bfsigma,\tau_0)$
% concretely in terms of a representative metric $g_{ab}\in \mathbf{g}$
% and a TT tensor $\sigma_{ab}$ with $[g_{ab},\sigma_{ab}]^*=\bfsigma$.
% We call a triple $(g_{ab},\sigma_{ab},\tau_0)$ \textbf{representative
% CMC conformal data}.

\begin{proposition}[The CMC Conformal Method]\label{prop:CMC}
Let $(\mathbf g, \bfsigma, \tau_0)$ be CMC conformal data.

Suppose $(g_{ab}, \sigma_{ab}, \tau_0)$ is an arbitrary
representative of the CMC conformal data (i.e., $[g_{ab}]=\mathbf{g}$
and $[g_{ab},\sigma_{ab}]^*=\bfsigma$)
and suppose 
$\phi$ is a positive function solving the Lichnerowicz-York equation \eqref{eq:lich}.
Then
\begin{equation}\label{eq:cmc-recipe}
\ifjournal\eqalign{
\ovl g_{ab} &= \phi^{q-2} g_{ab}\\
\ovl K_{ab} &= \phi^{-2}\sigma_{ab} + \frac{\tau_0}{n}\ovl{g}_{ab}	
}\else
\begin{aligned}
\ovl g_{ab} &= \phi^{q-2} g_{ab}\\
\ovl K_{ab} &= \phi^{-2}\sigma_{ab} + \frac{\tau_0}{n}\ovl{g}_{ab}
\end{aligned}\fi
\end{equation}
solve the constraint equations. Moreover,
\begin{equation}\label{eq:CMC-conds2}
\ifjournal\eqalign{
\strut[\ovl{g}_{ab}] &= \mathbf{g},\\
\strut[\ovl{g}_{ab}, \ovl K_{ab}]^* &= \bfsigma,\\
\ovl g^{ab} \ovl K_{ab} &= \tau_0.}
\else
\begin{aligned}
\strut[\ovl{g}_{ab}] &= \mathbf{g},\\
\strut[\ovl{g}_{ab}, \ovl K_{ab}]^* &= \bfsigma,\\
\ovl g^{ab} \ovl K_{ab} &= \tau_0.
\end{aligned}\fi
\end{equation}

Conversely, suppose $(\ovl g_{ab}, \ovl K_{ab})$ is a solution of the constraint
equations such that equations FOO\eqref{eq:CMC-conds2} are satisfied.
Let $(g_{ab}, \sigma_{ab}, \tau_0)$ be any representative of the CMC conformal data
and let $\phi$ be the unique conformal factor such that $\ovl g_{ab} = \phi^{q-2} g_{ab}$.  
Then $\phi$ solves the Lichnerowicz-York equation \eqref{eq:lich}.
\end{proposition}

As a consequence of Proposition \ref{prop:CMC}, the set of solutions of the constraints satisfying 
conditions \eqref{eq:CMC-conds2} is in one-to-one correspondence with the set of conformal factors $\phi$
solving the Lichnerowicz-York equation \eqref{eq:lich} as expressed with respect to any representative
of the CMC conformal data $(\mathbf{g},\bfsigma,\tau_0)$.  This independence with respect to
the choice of representative is known in the literature as \textbf{conformal covariance}.

\begin{proposition}\label{prop:CMCinv}
Suppose $(g_{ab}, \sigma_{ab}, \tau_0)$ and
$(\tilde g_{ab}, \tilde \sigma_{ab}, \tau_0)$ are two representatives of the same 
CMC conformal data, so $\tilde g_{ab} = \psi^{q-2} g_{ab}$ and $\tilde\sigma_{ab}=\psi^{-2}\sigma_{ab}$
for some conformal factor $\psi$.  Then $\phi$ solves the Lichnerowicz-York equation \eqref{eq:lich}
with respect to $(g_{ab}, \sigma_{ab}, \tau_0)$ if and only if $\psi^{-1}\phi$ solves
the Lichnerowicz-York equation \eqref{eq:lich} with respect to
$(\tilde g_{ab}, \tilde \sigma_{ab}, \tau_0)$, in which case the corresponding
solution $(\ovl g_{ab}, \ovl K_{ab})$ of the constraints in both cases is the same.
\end{proposition}
\begin{proof}
Suppose $\phi$ solves the Lichnerowicz-York equation with respect to 
$(g_{ab}, \sigma_{ab}, \tau_0)$ and let $(\ovl g_{ab},\ovl K_{ab})$ be
defined by equations \eqref{eq:cmc-recipe}.  The forward implication of Proposition \ref{prop:CMC}
implies $(\ovl g_{ab},\ovl K_{ab})$ is a solution of the constraints satisfying \eqref{eq:CMC-conds2}.

Since $\tilde g_{ab} = \psi^{q-2} g_{ab}$ and $\ovl g_{ab} = \phi^{q-2} g_{ab}$ it follows that
$\ovl g_{ab} = (\phi/\psi)^{q-2} \tilde g_{ab}$ and hence the reverse implication of Proposition
\ref{prop:CMC} implies that $\phi\psi^{-1}$ solves the Lichnerowicz-York equation \eqref{eq:lich}
with respect to  $(\tilde g_{ab}, \tilde \sigma_{ab}, \tau_0)$.  The solution of the constraints
generated by $\phi\psi^{-1}$ given by equation \eqref{eq:cmc-recipe} is
\begin{equation}
\ifjournal\eqalign{
\ovl{\tilde g}_{ab} &= (\phi/\psi)^{q-2} \tilde g_{ab} = \phi^{q-2} g_{ab} = \ovl g_{ab}\\
\ovl{\tilde K}_{ab} &= (\phi/\psi)^{-2} \tilde \sigma_{ab} + \frac{\tau_0}{n}\ovl{\tilde g}_{ab} 
= \phi^{-2}\sigma_{ab} + \frac{\tau_0} {n} \ovl g_{ab} = \ovl K_{ab}.	
}\else
\begin{aligned}
\ovl{\tilde g}_{ab} &= (\phi/\psi)^{q-2} \tilde g_{ab} = \phi^{q-2} g_{ab} = \ovl g_{ab}\\
\ovl{\tilde K}_{ab} &= (\phi/\psi)^{-2} \tilde \sigma_{ab} + \frac{\tau_0}{n}\ovl{\tilde g}_{ab} 
= \phi^{-2}\sigma_{ab} + \frac{\tau_0} {n} \ovl g_{ab} = \ovl K_{ab}.
\end{aligned}\fi
\end{equation}
\end{proof}

The celebrated property of the CMC conformal method is that given representative
CMC conformal data $(g_{ab},\sigma_{ab},\tau_0)$ there is (generically) 
exactly one solution of
the Lichnerowicz-York equation, so there is effectively a one-to-one correspondence between
CMC conformal data and CMC solutions of the constraints.  This result (accomplished over
many years by several authors including York and Choquet-Bruhat, and completed and 
summarized by Isenberg in \cite{Isenberg:1995bi}) can be expressed in terms of conformal
objects (independent of a choice of background metric) as follows.

\begin{theorem}[CMC Parameterization]\label{thm:CMC}
Let $(\mathbf{g},\bfsigma,\tau_0)$ be CMC conformal data.
Then there exists a unique solution $(\ovl g_{ab}, \ovl K_{ab})$
of the vacuum Einstein constraint equations satisfying conditions \eqref{eq:CMC-conds2}
except in the following cases:
\begin{itemize}
\item $\mathbf{g}$ is Yamabe positive and $\bfsigma=0$, in which case there is no solution,
\item $\mathbf{g}$ is Yamabe negative and $\tau_0=0$, in which case there is no solution,
\item $\mathbf{g}$ is Yamabe null and $\bfsigma=0$ or $\tau_0=0$,
in which case there is no solution (unless both are zero, in which case there is 
a one-parameter family of solutions consisting of solution metrics $\ovl g_{ab}$
all homothetically related to a single metric with vanishing scalar curvature
and with solution extrinsic curvatures $\ovl K_{ab}$ all vanishing identically).
\end{itemize}
\end{theorem}

\section{The 1974 Conformal Method}\label{sec:1974}

Let $\omega$ be a fixed
volume form, and suppose $(\ovl g_{ab}, \ovl K_{ab})$ is a solution of the constraint
equations.  The solution and the choice of $\omega$ uniquely determine the following:
\ifjournal
\numparts
\begin{eqnarray}\customlabel{eq:1974params}{\arabic{section}.\arabic{eqnval}}
\mathbf{g}&=[\ovl g_{ab}], \\
\bfsigma &= P_{\omega}( [\ovl g_{ab}, \ovl K_{ab}]^*),\label{eq:1974sigma}\\
\tau &= \ovl g^{ab} \ovl K_{ab}.
\end{eqnarray}
\endnumparts
\else
\begin{subequations}\label{eq:1974params}
\begin{alignat}{2}
\mathbf{g}&=[\ovl g_{ab}], \\
\bfsigma &= P_{\omega}( [\ovl g_{ab}, \ovl K_{ab}]^*),\label{eq:1974sigma}\\
\tau &= \ovl g^{ab} \ovl K_{ab}.
\end{alignat}
\end{subequations}
\fi
where $[\ovl g_{ab}, \ovl K_{ab}]^*$ is defined at the end of Definition \ref{def:cotangent}
and the projection $P_\omega$ comes from Definition \ref{def:proj}.

We call a tuple $(\mathbf{g},\bfsigma,\tau,\omega)$
\textbf{1974 conformal data}.   
Although it is not usually presented this way, the 1974 conformal method
attempts to reverse this process: starting from conformal data
$(\mathbf{g},\bfsigma,\tau,\omega)$ 
we seek a solution $(\ovl g_{ab}, \ovl K_{ab})$
of the constraints satisfying conditions \eqref{eq:1974params}.  

Suppose $(\ovl g_{ab}, \ovl K_{ab})$ is a solution of the constraints
satisfying conditions \eqref{eq:1974params}
and let $\ovl A_{ab}$ be the trace-free part of $\ovl K_{ab}$, so
$K_{ab} = A_{ab} + (\tau/n) \ovl g_{ab}$.
Equation \eqref{eq:1974sigma} is equivalent to the existence of a vector field $W^a$
such that
\begin{equation}\label{eq:olAsplit}
[\ovl g_{ab}, \ovl A_{ab}]^* = \bfsigma + k_\omega(\ck_{\bf g} W).
\end{equation}
Let $g_{ab}$ be the unique element of $\mathbf{g}$ with
$dV_g = \omega$, and let $A_{ab}$ be the representative of $[\ovl g_{ab}, \ovl A_{ab}]^*$
with respect to $g_{ab}$ (i.e., $\ovl A_{ab} = \phi^{-2} A_{ab}$). 
From our specific choice of $g_{ab}$, Proposition
\ref{prop:yorksplitconf} implies that equation \eqref{eq:olAsplit}
is equivalent to
\begin{equation}
A_{ab} = \sigma_{ab} + (\ck_{g} W)_{ab}
\end{equation}
where $\sigma_{ab}$ is the representative of $\bfsigma$ with respect
to $g_{ab}$. We then have
\ifjournal
\numparts
\begin{eqnarray}
\customlabel{eq:1974decomp}{\arabic{section}.\arabic{eqnval}}
\ovl g_{ab} &= \phi^{q-2} g_{ab}\\
\ovl K_{ab} &= \phi^{-2}\left[\sigma_{ab } + (\ck_{g} W)_{ab}\right] + \frac{\tau}{n} \ovl g_{ab}.\label{eq:1974Kdecom}
\end{eqnarray}
\endnumparts
\else
\begin{subequations}\label{eq:1974decomp}
\begin{alignat}{2}
\ovl g_{ab} &= \phi^{q-2} g_{ab}\\
\ovl K_{ab} &= \phi^{-2}\left[\sigma_{ab } + (\ck_{g} W)_{ab}\right] + \frac{\tau}{n} \ovl g_{ab}.\label{eq:1974Kdecom}
\end{alignat}
\end{subequations}
\fi
The preceding discussion is reversible, so we have shown that equations \eqref{eq:1974params}
are equivalent to the existence a conformal factor $\phi$ and a vector field $W^a$ such
that conditions $\eqref{eq:1974decomp}$ hold, \textbf{so long as} $g_{ab}$ is the representative 
of $\mathbf{g}$ with $dV_g = \omega$.

Substituting equations \eqref{eq:1974decomp}
into the constraint equations \eqref{eq:constraints} we find that $(\ovl g_{ab},\ovl K_{ab})$ is a solution
of the constraints if and only if $\phi$ and $W$ satisfy
\ifjournal
\numparts
\begin{eqnarray}
\customlabel{eq:LCBY}{\arabic{section}.\arabic{eqnval}}
\hskip-13ex -2\kappa q\,\Lap \phi + R_g \phi -\left|\sigma+\ck_g W\right|_g^2\phi^{-q-1} + \kappa \tau^2\phi^{q-1} &= 0 &\quad\text{\small[1974 Hamiltonian constraint]}\label{eq:LCBYHamiltonian}\\
\div_g \ck W &= \kappa\phi^q d\tau.
 &\quad\text{\small[1974 momentum constraint]}\label{eq:LCBYmomentum}
\end{eqnarray}
\endnumparts
\else
\begin{subequations}\label{eq:LCBY}
\begin{alignat}{2}
-2\kappa q\,\Lap \phi + R_g \phi -\left|\sigma+\ck_g W\right|_g^2\phi^{-q-1} + \kappa \tau^2\phi^{q-1} &= 0 &\qquad&\text{\small[1974 Hamiltonian constraint]}\label{eq:LCBYHamiltonian}\\
\div_g \ck W &= \kappa\phi^q d\tau.
 &\qquad&\text{\small[1974 momentum constraint]}\label{eq:LCBYmomentum}
\end{alignat}
\end{subequations}
\fi
These equations, which first appeared in \cite{OMurchadha:1974bf}, will be called 
the \textbf{1974 conformally parameterized constraint equations}, though we note
that they have various other names in the literature, including the LCBY equations named after
Lichnerowicz, Choquet-Bruhat and York.  We have described how their solutions correspond to the 
solutions of the constraints solving conditions \eqref{eq:1974params}, and summarize this
discussion as follows.
\begin{proposition}[1974 Conformal Method]
Let $(\mathbf{g},\bfsigma,\tau,\omega)$ be 1974 conformal data.

Let $g_{ab}\in\mathbf{g}$ be the unique representative
with $dV_g=\omega$ and let $\sigma_{ab}$ be the representative
of $\bfsigma$ with respect to $g_{ab}$.

Suppose $\phi$ and  $W^a$ solve the 1974 conformally parameterized constraint equations \eqref{eq:LCBY}
with respect to $g_{ab}$ and $\sigma_{ab}$.
Then $(\ovl g_{ab}, \ovl K_{ab})$ defined by equations \eqref{eq:1974decomp} 
satisfy the constraint equations \eqref{eq:constraints} and satisfy
\ifjournal
\numparts
\begin{eqnarray}
\customlabel{eq:1974params2}{\arabic{section}.\arabic{eqnval}}
[\ovl g_{ab}]&=\mathbf{g}, \\
P_{\omega}( [\ovl g_{ab}, \ovl K_{ab}]^*)&=\bfsigma ,\label{eq:1974sigma2}\\
\ovl g^{ab} \ovl K_{ab} &= \tau.
\end{eqnarray}
\endnumparts
\else
\begin{subequations}\label{eq:1974params2}
\begin{alignat}{2}
[\ovl g_{ab}]&=\mathbf{g}, \\
P_{\omega}( [\ovl g_{ab}, \ovl K_{ab}]^*)&=\bfsigma ,\label{eq:1974sigma2}\\
\ovl g^{ab} \ovl K_{ab} &= \tau.
\end{alignat}
\end{subequations}
\fi

Conversely, suppose $(\ovl g_{ab}, \ovl K_{ab})$ is a solution of the constraints 
satisfying conditions \eqref{eq:1974params2}. 
Then there exists a conformal factor $\phi$ and vector field
$W^a$ (both unique up to addition of a conformal Killing field to $W^a$)
such that the decomposition \eqref{eq:1974decomp} holds and
the 1974 conformally parameterized constraint equations (with respect to $g_{ab}$ and $\sigma_{ab}$) are satisfied.
\end{proposition}

Each choice of volume form $\omega$ leads to an independent parameterization of the
set of solutions of the constraints in the sense that once $\omega$ is fixed,
every solution $(\ovl g_{ab}, \ovl K_{ab})$ is associated with exactly one tuple
$(\mathbf{g}, \bfsigma, \tau, \omega)$ via equations \eqref{eq:1974params2}.  
The reverse implication need not be true, however. 
As mentioned in the introduction, for 1974 conformal data where $\tau$
is not nearly constant it is generally unknown how many solutions of the constraints
are associated with this data.

In the usual presentation of the 1974 conformal method
the representative conformal data consist of a metric $g_{ab}$,
a TT tensor $\sigma_{ab}$, and a mean curvature $\tau$,
and we begin by writing down the corresponding 1974 conformally parameterized constraint equations.
The triple $(g_{ab}, \sigma_{ab}, \tau)$ appears to be analogous
to representative data for the CMC conformal method, but
representative data determines 
conformal data $(\mathbf{g}, \bfsigma, \tau, \omega)$
as follows:
\begin{equation}
\ifjournal\eqalign{
\mathbf{g} &= [ g_{ab}]\\
\bfsigma &= [g_{ab}, \sigma_{ab}]^*\\
\tau &= \tau\\
\omega &= dV_g.	
}\else
\begin{aligned}
\mathbf{g} &= [ g_{ab}]\\
\bfsigma &= [g_{ab}, \sigma_{ab}]^*\\
\tau &= \tau\\
\omega &= dV_g.
\end{aligned}\fi
\end{equation}
Note that compared to the CMC conformal method,
the choice of metric now plays two roles: it selects the conformal
class of the solution metric and the choice of volume form $\omega$ in
Proposition \ref{prop:yorksplitconf}.  In this second role, it determines
a choice of identification of $T_{\mathbf{g}}\calC$ with
$T^*_\mathbf{g}\calC$.
% , and to make the connection between the 
% 1974 method and the CTS method it will be useful to decouple these two roles 
% and make the choice of $\omega$ explicit.
Although the choice of volume form $\omega$ and the choice of background
metric $g_{ab}$ used to write down the PDEs \eqref{eq:LCBY} are tightly connected 
in the 1974 conformal method, there is no particular reason why this needs to
be the case.  Indeed, there are good reasons to decouple these two roles.  Given an
arbitrary $\omega$, one might want to work with a metric different from the one for
which $dV_{g} = \omega$; it may be more expedient to work with a metric with, e.g.,
positive scalar curvature instead.  The problem of finding a solution of the 
constraint equations satisfying conditions \eqref{eq:1974params2} does not depend 
on the choice of a background metric.  But the 1974 conformally parameterized constraint equations themselves do
depend on the choice $dV_g=\omega$.  If we work with a different background metric,
these equations will change, and we will see that this is the connection between
the 1974 conformal method and the Hamiltonian formulation of the CTS method.

When expressed in terms of representative conformal data,
the 1974 method appears to lack conformal covariance. 
If we start with representative data 
$(g_{ab}, \sigma_{ab}, \tau)$ and conformally change to representative data
$(\phi^{q-2} g_{ab}, \phi^{-2}\sigma_{ab}, \tau)$, there is no
reason to expect that the corresponding solutions of the constraints will be the same.
In terms of conformal objects, this transformation is equivalent to
fixing $(\mathbf{g},\bfsigma,\tau)$ but changing the choice of $\omega$.
Each choice of $\omega$ gives a separate parameterization of the solutions
of the constraint equations, and the parameterizations can be different from each other.
Indeed, recent work \cite{Maxwell:2014b}
shows that the 1974 conformal method parameterizes flat Kasner data in fundamentally different ways depending
on the choice of volume form. Certain data $(\bf{g},{\bfsigma},\tau)$ generate one-parameter
families of solutions for some volume forms, but generate only a single
solution for others.  So the choice of $\omega$ is an important part 
of the parameterization.  However, the task defined by the 1974 conformal method,
$$
\text{Find a solution of the constraints satisfying conditions \eqref{eq:1974params2}}.
$$
can be expressed in terms of conformal objects and therefore is by necessity conformally
covariant; the issue is simply how to express the problem when using a representative
metric different from the one with $dV_g = \omega$.

Finally, we observe that although the choice of $\omega$ is important for the 1974 conformal
method, if we restrict to constant mean curvature data $\tau = \tau_0$, 
then the choice of $\omega$ is irrelevant. The 1974 
conformally parameterized momentum constraint  \eqref{eq:LCBYmomentum} reads
\begin{equation}
\div_g \ck W = 0
\end{equation}
which is solved exactly by conformal Killing fields (i.e., $(\ck W)_{ab} = 0$). So there is no
longer any ambiguity about adding tangent and cotangent vectors in the expression
\begin{equation}
A_{ab} = \sigma_{ab} + (\ck W)_{ab}
\end{equation}
from equation \eqref{eq:1974Kdecom} and the choice of volume form is no longer needed.
The 1974 conformally parameterized Hamiltonian
constraint  \eqref{eq:LCBYHamiltonian} is
\begin{equation}
-2\kappa q\,\Lap \phi + R_g \phi -\left|\sigma\right|_g^2\phi^{-q-1} + \kappa \tau^2_0\phi^{q-1} = 0,
\end{equation}
i.e., the Lichnerowicz-York equation \eqref{eq:lich}, so the 1974 method reduces to the CMC-conformal method.

% Maybe the following is true:
% \begin{proposition}
% If $L_g W = \phi^q L_g \hat W$ then $\phi$ is constant
%on every open set on which $L_g\neq 0$.
% \end{proposition}

\section{The Conformal Thin-Sandwich Method}\label{sec:CTS}

The thick-sandwich conjecture, in the vacuum setting, states that given two metrics
$\ovl g_{ab}^{0}$ and $\ovl g_{ab}^{1}$ on $M^n$ one can find a
globally hyperbolic 
Ricci-flat Lorentzian spacetime, unique up to diffeomorphism, 
and two disjoint 
spacelike hypersurfaces of the spacetime, such that the induced metrics
on the hypersurfaces are the given two metrics. As described
in \cite{Bartnik:1993jl}, there are reasons to doubt the validity of
this conjecture.  It was also shown in \cite{Bartnik:1993jl}
that an infinitesimal variation, known as the thin-sandwich conjecture, 
turns out to hold under limited circumstances.  The 
thin-sandwich conjecture asserts that given $\ovl g_{ab}$ and its 
Lie derivative $\dot{\ovl g}_{ab} = \calL_T g_{ab}$ with
respect to some (to be determined) future pointing time like vector field 
$T^a$ along the surface, there is a unique Ricci-flat spacetime containing
a slice satisfying the initial conditions. Writing $T= N\nu^a + X^a$ where 
$\nu^a$ is the future pointing unit normal to the surface and $X^a$ is a vector field
tangential to the surface (i.e., in terms of the lapse $N$ and shift $X^a$) we have
\begin{equation}\label{eq:dgdt}
\dot{\overline{g}}_{ab} = 2\ovl N\, \ovl K_{ab} + \calL_X \ovl g_{ab}.
\end{equation}
So the goal is to find $(\ovl g_{ab}, \ovl K_{ab})$ solving the constraints,
along with a choice of $\ovl N$ and $X^a$, such that $\ovl g_{ab}$ is 
the given metric and such that \eqref{eq:dgdt} holds.
Using a perturbative technique, the authors of \cite{Bartnik:1993jl} 
exhibited an open set of data $(\ovl g_{ab}, \dot{\ovl g}_{ab})$, with
additional restrictions on the scalar curvature of $\ovl g_{ab}$, for which
the conjecture is true.  They also conjecture, however, that the thin-sandwich conjecture
is not well-posed in general.

York's conformal thin-sandwich method \cite{YorkJr:1999jo} is based on
a conformal version of the thin-sandwich conjecture. Given a conformal class $\mathbf{g}$ and
a conformal tangent vector $\mathbf{u}\in T_{\mathbf{g}} \calC$, 
we wish to find a solution $(\ovl{g}_{ab},\ovl{K}_{ab})$ of 
the vacuum Einstein constraint equations along with a lapse $\ovl{N}$ and a shift $\ovl{X}^a$
such that
\ifjournal\begin{eqnarray}
[\ovl{g}_{ab}] &= \mathbf{g},\label{eq:cts1}\\
\left[\ovl{g}_{ab},\dot{\ovl{g}}_{ab}\right] &= \mathbf{u}, \label{eq:proj}
\end{eqnarray}
\else
\begin{align}
[\ovl{g}_{ab}] &= \mathbf{g},\label{eq:cts1}\\
\left[\ovl{g}_{ab},\dot{\ovl{g}}_{ab}\right] &= \mathbf{u}, \label{eq:proj}
\end{align}
\fi
where $\dot{\ovl{g}}_{ab}$ is defined by equation \eqref{eq:dgdt} and 
(as noted at the end of  Definition \ref{def:tangent})
$\left[\ovl{g}_{ab},\dot{\ovl{g}}_{ab}\right]$ should be thought of as the 
pushforward of the tangent vector $\dot{\ovl g}_{ab}$
to an element of $T_{\mathbf{g}} \calC$.  One hopes that specification of
$(\mathbf{g}, \mathbf{u})$,  along with information about the trace part of $\ovl K_{ab}$
and the coordinate freedom in $(\ovl N, X^a)$, results in a unique solution of the constraint equations.

Let $\mathbf{g}\in\calC$ and $\mathbf{u}\in T_{\mathbf g}\calC$ be given and
suppose $(\ovl g_{ab}, \ovl K_{ab}, \ovl N, X^a)$ satisfy equations 
\eqref{eq:cts1}, and \eqref{eq:proj} (with $\dot{\ovl g}_{ab}$
defined by equation \eqref{eq:dgdt}).  Let 
$\tau$ be the trace of $\ovl K_{ab}$ so
\begin{equation}
\ovl K_{ab} = \ovl A_{ab} + \frac{\tau}{n} \ovl g_{ab}
\end{equation}
for some unique trace-free tensor $\ovl A_{ab}$.
Decomposing equation \eqref{eq:dgdt} into its trace-free and trace parts we 
obtain
\begin{equation}
\dot{\overline{g}}_{ab}
= \left(2\ovl N\, \ovl A_{ab} + (\ck_{\ovl g} X)_{ab}\right) + \left(\ovl N\,\tau + \div_{\ovl g} X\right)\frac{2}{n} \ovl g_{ab},
\end{equation}
so equation \eqref{eq:proj} is equivalent to
\begin{equation}\label{eq:projcoord}
2 \ovl N\, \ovl A_{ab} + (\ck_{\ovl g} X)_{ab} = \ovl u_{ab}
\end{equation}
where $\ovl u_{ab}$ is the representative of $\mathbf u$ with respect to $\ovl g_{ab}$.
Equation \eqref{eq:projcoord} can be solved for $\ovl A_{ab}$ to obtain
\begin{equation}\label{eq:CTSA}
\ovl A_{ab} = \frac{1}{2\ovl N}\left[\ovl u_{ab} - (\ck_{\ovl g} X)_{ab}\right],
\end{equation}
and the constraint equations \eqref{eq:constraints}
can be written in terms of $\ovl A_{ab}$ and $\tau$ to obtain
\begin{equation}\label{eq:tfconstraints}
\ifjournal\eqalign{
R_{\ovl g} - \left|\ovl A\right|_{\ovl g}^2 + \kappa \tau ^2 &= 0\\
\div_{\ovl g} \ovl A &= \kappa d \tau.	
}\else
\begin{aligned}
R_{\ovl g} - \left|\ovl A\right|_{\ovl g}^2 + \kappa \tau ^2 &= 0\\
\div_{\ovl g} \ovl A & = \kappa d \tau.
\end{aligned}\fi
\end{equation}
Substituting equation\eqref{eq:CTSA} into equations \eqref{eq:tfconstraints} we have 
\ifjournal
\numparts
\begin{eqnarray}
\customlabel{eq:conftrans}{\arabic{section}.\arabic{eqnval}}
R_{\ovl g} - \left|\frac{1}{2\ovl N}\left[\ovl u - \ck_{\ovl g} X\right]\right|^2_{\ovl g} + \kappa \tau ^2 &= 0
\\
\div_{\ovl g} \left[ \frac{1}{2\ovl N}\left[\ovl u - \ck_{\ovl g} X\right]\right] & = \kappa d \tau.
\end{eqnarray}
\endnumparts
\else
\begin{subequations}
\label{eq:conftrans}
\begin{alignat}{2}
R_{\ovl g} - \left|\frac{1}{2\ovl N}\left[\ovl u - \ck_{\ovl g} X\right]\right|^2_{\ovl g} + \kappa \tau ^2 &= 0
\\
\div_{\ovl g} \left[ \frac{1}{2\ovl N}\left[\ovl u - \ck_{\ovl g} X\right]\right] & = \kappa d \tau.
\end{alignat}
\end{subequations}\fi

York's prescription for solving these equations can be described as follows. Pick
an arbitrary $g_{ab}\in \mathbf{g}$ and let $u_{ab}$ be the representative of $\mathbf{u}$
with respect to $g_{ab}$. 
The solution metric $\ovl g_{ab}$ is then related to to $g_{ab}$ via an as-yet unknown conformal factor 
$\phi$ via $\ovl g_{ab} = \phi^{q-2} g_{ab}$, and we set $\ovl u_{ab} = \phi^{q-2} u_{ab}$
so that $[\ovl g_{ab}, \ovl u_{ab}]=[g_{ab}, u_{ab}] = \mathbf{u}$.
The shift $X^a$ is the other unknown, and the remaining quantities $\ovl N$ and $\tau$
are treated as parameters. The mean curvature $\tau$ is specified directly,
but the lapse obeys a nontrivial conformal transformation law: $\ovl N = \phi^q N$, where $N$ is a given 
positive function.  

Rewriting equations \eqref{eq:conftrans} in terms of $\phi$, $X$, $g_{ab}$, $u_{ab}$, $\tau$,
and $N$, we obtain the CTS equations
\ifjournal
\numparts
\begin{eqnarray}
\customlabel{eq:CTS}{\arabic{section}.\arabic{eqnval}}
\hskip-14ex-2\kappa q\Delta \phi +R_g\phi- \left| \frac{1}{2N}(u-\ck_g X) \right|^2_g\phi^{-q-1} + \kappa \tau^2\phi^{q-1} = 0
&\hskip0.5ex\text{\small[CTS Hamiltonian constraint]}\\
-\div_g\left[ \frac{1}{2N}\ck_g X\right] = -\div_g\left[ \frac{1}{2N} u\right] + \kappa\phi^q\;d\tau
&\hskip0.5ex\text{\small[CTS momentum constraint]}
\end{eqnarray}
\endnumparts
\else
\begin{subequations}\label{eq:CTS}
\begin{alignat}{2}
-2\kappa q\Delta \phi +R_g\phi- \left| \frac{1}{2N}(u-\ck_g X) \right|^2_g\phi^{-q-1} + \kappa \tau^2\phi^{q-1} &= 0
&\quad\text{\small[CTS Hamiltonian constraint]}\\
-\div_g\left[ \frac{1}{2N}\ck_g X\right] &= -\div_g\left[ \frac{1}{2N} u\right] + \kappa\phi^q\;d\tau
&\quad\text{\small[CTS momentum constraint]}
\end{alignat}
\end{subequations}\fi
to be solved for $\phi$ and $X^a$.

The conformally transforming lapse is the key ingredient of the conformal thin-sandwich
method, and can be motivated by examining the term $\ovl u - \ck_{\ovl g} X$ appearing
in the momentum constraint of equation \eqref{eq:conftrans}.  This term represents a conformal tangent vector, e.g.
\begin{equation}\label{eq:divarg-tangent}
\ovl u_{ab} - (\ck_{\ovl g} X)_{ab} = \phi^{q-2}\left[u_{ab} - (\ck_{g} X)_{ab}\right].
\end{equation}
The divergence, however, naturally acts on conformal cotangent vectors, so we should have
\begin{equation}\label{eq:divarg-cotangent}
\frac{1}{2\ovl N}\left[\ovl u_{ab} - (\ck_{\ovl g} X)_{ab}\right] = 
\phi^{-2} \frac{1}{2 N}\left[u_{ab} - (\ck_{g} X)_{ab}\right].
\end{equation}
Comparing equations \eqref{eq:divarg-tangent} and \eqref{eq:divarg-cotangent}
we arrive at York's transformation law $\ovl N = \phi^{q} N$.  

A conformally transforming
lapse is a conformal object associated with a conformal class $\mathbf{g}$, and it
will be useful to introduce the following notation.
\begin{definition}
A \textbf{densitized lapse} is an element of $(\calM \times C^\infty_+(M))/\sim$ where
\begin{equation}
(\tilde g_{ab},\tilde N) \sim (g_{ab}, N)
\end{equation}
if there is a smooth positive function $\phi$ on $M$ with $\tilde g_{ab} = \phi^{q-2} g_{ab}$
and $\tilde N = \phi^q N$.  We use the following notation:
\begin{itemize}
\item $[g_{ab}, N]$ is the densitized lapse determined by $(g_{ab},N)$,
\item $\calN_{\mathbf{g}}$ is the set of all densitized lapses $[g_{ab},N]$ with $g_{ab}\in \mathbf{g}$,
\item $\calN = \bigcup_{\mathbf{g}\in \calC} \calN_{\mathbf{g}}$.
\end{itemize}
\end{definition}

A tuple $(\mathbf{g},\mathbf{u}, \tau, \mathbf{N})$ 
where $\mathbf{g}\in\calC$, $\mathbf u \in T_{\mathbf g} \calC$ and $\mathbf N \in \calN_\mathbf g$
is called \textbf{CTS data}, and
$(g_{ab}, u_{ab}, \tau, N)$ is a \textbf{representative} if 
$[g_{ab}]=\mathbf g$, $[g_{ab}, u_{ab}] = \mathbf u$ and $[g_{ab}, N] = \mathbf N$.
With this notation, we summarize the previous discussion as follows.

\begin{proposition}[The CTS Method]\label{prop:CTS}
Let $(\mathbf g, \mathbf u, \tau, \mathbf N)$ be CTS data.

Suppose $(g_{ab}, u_{ab}, \tau, N)$ is a 
representative of the CTS data. If $\phi$ and $X^a$ solve the CTS equations \eqref{eq:CTS}
with respect to $(g_{ab}, u_{ab}, \tau, N)$, then
\begin{equation}\label{eq:ctsl-recipe}
\ifjournal\eqalign{
\ovl g_{ab} &= \phi^q g_{ab}\\
\ovl K_{ab} &= \phi^{-2}\frac{1}{2N}\left[u_{ab} - (\ck X)_{ab} \right] + \frac{\tau}{n}\ovl{g}_{ab}	
}\else
\begin{aligned}
\ovl g_{ab} &= \phi^q g_{ab}\\
\ovl K_{ab} &= \phi^{-2}\frac{1}{2N}\left[u_{ab} - (\ck X)_{ab} \right] + \frac{\tau}{n}\ovl{g}_{ab}
\end{aligned}\fi
\end{equation}
solve the constraint equations.  Moreover, setting $\ovl N = \phi^q N$
and
\begin{equation}
\dot{\ovl{g}}_{ab} = 2\ovl{N}\, \ovl K_{ab} + \Lie_X \ovl{g}_{ab},
\end{equation}
we have
\begin{equation}\label{eq:CTS-Lconds}
\ifjournal\eqalign{
\strut[\ovl{g}] &= \mathbf{g},\\
\strut[\ovl{g}, \dot{\ovl g}] &= \mathbf{u},\\
\ovl g^{ab} \ovl K_{ab} &= \tau,\\
\strut[\ovl g, \ovl N]&=\mathbf{N}.	
}\else
\begin{aligned}
\strut[\ovl{g}] &= \mathbf{g},\\
\strut[\ovl{g}, \dot{\ovl g}] &= \mathbf{u},\\
\ovl g^{ab} \ovl K_{ab} &= \tau,\\
\strut[\ovl g, \ovl N]&=\mathbf{N}.
\end{aligned}\fi
\end{equation}

Conversely, suppose $(\ovl g_{ab}, \ovl K_{ab})$ are solutions of the constraint equations
\eqref{eq:constraints} and that $\ovl N$ 
and $X^a$ are are a lapse and a shift such that conditions \eqref{eq:CTS-Lconds} hold.  
Let $(g_{ab}, u_{ab}, \tau, N)$ be any representative CTS data
for $(\mathbf{g}, \mathbf{u}, \tau, \mathbf{N})$ and let $\phi$ be the
unique conformal factor such that $\ovl g_{ab}=\phi^{q-2} g_{ab}$.  
Then $(\phi, X^a)$ solve the CTS equations \eqref{eq:CTS}
with respect to $(g_{ab}, u_{ab}, \tau, N)$ and equations \eqref{eq:ctsl-recipe} hold.
\end{proposition}

The CTS method is conformally
covariant in the sense that if we change to a second background metric, and conformally
transform the remaining representative conformal data to represent the same conformal objects, 
the resulting set of solutions of the constraint equations are the same.  
\begin{proposition}\label{prop:CTSinv}
Let $(g_{ab}, u_{ab}, \tau, N)$ and $(\tilde g_{ab}, \tilde u_{ab}, \tau, \tilde N)$
be representative CTS data both corresponding to the same CTS data 
$(\mathbf{g}, \mathbf{u}, \tau, \mathbf{N})$, and let $\psi$ be the unique conformal
factor such that $\tilde g_{ab} = \psi^{q-2} g_{ab}$.  
Then $(\phi,X^a)$ solves the CTS equations \eqref{eq:CTS} with respect to $(g_{ab}, u_{ab}, \tau, N)$
if and only if $(\psi^{-1}\phi, X^a)$ solves the CTS equations with respect to
$(\tilde{g}_{ab}, \tilde{u}_{ab}, \tau, \tilde N)$, and the corresponding solution $(\ovl g_{ab}, \ovl K_{ab})$
of the Einstein constraint equations in both cases is the same.
\end{proposition}
\begin{proof}
The proof is analogous to that of Proposition \ref{prop:CMCinv}.
\end{proof}

Each choice of densitized lapse yields an independent parameterization of the set of solutions of the
constraint equations in the sense that once the densitized lapse $\mathbf N$
is fixed, each solution of the constraints is associated with a tuple of CTS data
$(\mathbf g, \mathbf u, \tau, \mathbf N)$, and this data is unique up to adding an element of
$\Im\ck_\mathbf g$ to $\mathbf u$.

\begin{proposition}\label{prop:CTSparam}
Let $\mathbf{g}$ be a conformal class and let $\mathbf{N}\in \calN_{\mathbf{g}}$.  
Suppose $(\ovl{g}_{ab}, \ovl{K}_{ab})$ is a solution of the constraints
with $\ovl{g}_{ab}\in\mathbf{g}$.  Then $(\ovl{g}_{ab}, \ovl{K}_{ab})$ is
generated by CTS data $(\mathbf g, \mathbf u, \tau, \mathbf N)$ if and only if
$\tau = \ovl g^{ab} \ovl K_{ab}$ and $\mathbf u \in [\ovl g_{ab}, 2\ovl N\, \ovl K_{ab}] + \Im \ck_{\mathbf g}$,
where $\ovl N$ is the representative of $\mathbf N$ with respect to $\ovl g_{ab}$.
\end{proposition}
\begin{proof}
Let $(\ovl g_{ab}, \ovl K_{ab})$ be a solution of the constraints with $\ovl g_{ab}\in\mathbf{g}$, and let 
$\ovl N$ be the representative of $\mathbf{N}$ with respect to $\ovl g_{ab}$.
From Proposition \ref{prop:CTS} we see that 
$(\mathbf g, \mathbf u, \tau, \mathbf N)$ generates $(\ovl g_{ab}, \ovl K_{ab})$ if and only if
$\tau = \ovl g^{ab} \ovl K_{ab}$ and 
there is a vector field $X^a$ such that
\begin{equation}
\dot{\ovl g}_{ab} = 2\ovl N \,\ovl K_{ab} + \calL_{X} \ovl g_{ab}
\end{equation}
satisfies
\begin{equation}
[\ovl g_{ab}, \dot{\ovl g}_{ab}] = \mathbf{u}.
\end{equation}
Given a vector field $X^a$,
\begin{equation}
\ifjournal\eqalign{
[\ovl g_{ab}, 2\ovl N\, \ovl K_{ab} + \calL_{X}\ovl g_{ab}] 
&=[\ovl g_{ab}, 2\ovl N\, \ovl K_{ab}] + [\ovl g_{ab}, \calL_{X} \ovl g_{ab}]\\
&=[\ovl g_{ab}, 2\ovl N\, \ovl K_{ab}] + [\ovl g_{ab}, \ck_{\ovl g} {X}]\\
&=[\ovl g_{ab}, 2\ovl N\, \ovl K_{ab}] + \ck_{\mathbf g} X.	
}\else
\begin{aligned}
{}
[\ovl g_{ab}, 2\ovl N\, \ovl K_{ab} + \calL_{X}\ovl g_{ab}] 
&=[\ovl g_{ab}, 2\ovl N\, \ovl K_{ab}] + [\ovl g_{ab}, \calL_{X} \ovl g_{ab}]\\
&=[\ovl g_{ab}, 2\ovl N\, \ovl K_{ab}] + [\ovl g_{ab}, \ck_{\ovl g} {X}]\\
&=[\ovl g_{ab}, 2\ovl N\, \ovl K_{ab}] + \ck_{\mathbf g} X.
\end{aligned}\fi
\end{equation}
Thus $(\mathbf g, \mathbf u, \tau, \mathbf N)$ generates $(\ovl g_{ab}, \ovl K_{ab})$ 
if and only if $\tau = \ovl g^{ab} \ovl K_{ab}$ and
\begin{equation}
\mathbf u \in [\ovl g_{ab}, 2\ovl N\, \ovl K_{ab}] + \Im \ck_{\mathbf g}.
\end{equation}
\end{proof}

Proposition \ref{prop:CTSparam} shows that the true parameters for the
CTS method are a conformal class $\mathbf{g}$, a mean curvature $\tau$, 
a densitized lapse $\mathbf N$, and an element
of
$
T_{\mathbf{g}} \calC / \Im \ck_{\mathbf{g}},
$
i.e., an element of $T_{\mathbf g}(\calC/\calD_0)$ 
from Definition \ref{def:cgv}.  After selecting a densitized
lapse $\mathbf N$, a tuple
\begin{equation}
(\mathbf g, \mathbf u + \Im \ck_{\mathbf g}, \tau, \mathbf N)
\end{equation}
of \textbf{geometric CTS data} is uniquely determined by 
a solution $(\ovl g_{ab}, \ovl K_{ab})$ of the constraints, and
the CTS method attempts to invert this map.

\section{The Hamiltonian Formulation of the CTS Method}\label{sec:CTSH}
The CTS method was presented by York as a Lagrangian alternative to
the standard (Hamiltonian) conformal method.  Subsequently
Pfeiffer and York demonstrated 
a Hamiltonian approach \cite{Pfeiffer:2003ka} to the CTS method that will allow
us to link the CTS method and to the 1974 conformal method.
We will call the method described here the CTS-H method
(and will call the original conformal thin-sandwich approach
the CTS-L method if emphasis on its Lagrangian nature is desired).

Although not presented this way in \cite{Pfeiffer:2003ka}, the
key to the CTS-H method is the introduction of a lapse-dependent
way of translating between conformal tangent vectors and smooth
conformal cotangent vectors defined as follows.

\begin{definition}\label{def:iotaN}
Let $\mathbf{g}$ be a conformal class and let $\mathbf{N}$ be a 
densitized lapse.  Given a conformal velocity $\mathbf{u}\in T_{\mathbf{g}}\calC$,
we wish to identify it with an element of $T^*_{\mathbf{g}}\calC$.  To do this,
let $g_{ab}$ be an arbitrary representative of $\mathbf{g}$ and let
$u_{ab}$ and $N$ be the representatives of $\mathbf{u}$ and $\mathbf{N}$ 
with respect to $g_{ab}$. We then define
\begin{equation}\label{eq:iotaN}
k_\mathbf{N}(\mathbf u) = [ g_{ab}, (1/2N) u_{ab}]^*.
\end{equation}
\end{definition}

To see that $k_\mathbf{N}$ is well-defined, suppose we
use a different representative metric $\tilde g_{ab} = \phi^{q-2} g_{ab}$.
Then $\tilde u_{ab} = \phi^{q-2} u_{ab}$ and $\tilde N = \phi^q N$ so
\begin{equation}
\frac{1}{2\tilde N} \tilde u_{ab} = \frac{1}{2\phi^q N} \phi^{q-2} u_{ab} = \phi^{-2}\frac{1}{2 N}  u_{ab}.
\end{equation}
Hence
\begin{equation}
[ \tilde g_{ab}, (1/2\tilde N) \tilde u_{ab}]^*
= [ \phi^q g_{ab}, \phi^{-2}(1/2 N)  u_{ab}]^* 
= [ g_{ab}, (1/2 N)  u_{ab}]^*
\end{equation}
as needed.

The map $k_\mathbf{N}$ plays the same role
for the CTS-H method as $k_{\omega}$ 
does for the 1974
conformal method, and 
in fact there is a way to identify densitized lapses
and volume forms such that the corresponding maps 
$k_{\mathbf{N}}$ and $k_{\omega}$
are identical.

\begin{proposition}\label{prop:Nisomega}
Suppose $\mathbf{N} \in \calN_{\mathbf{g}}$ for some
conformal class $\mathbf{g}$.  Let $g_{ab}$ be an arbitrary
representative of $\mathbf{g}$, let $N$ be the
representative of $\mathbf{N}$ with respect to $g_{ab}$,
and let
\begin{equation}\label{eq:omegadef}
\omega = \frac{1}{2N} dV_g.
\end{equation}
Then
\begin{equation}\label{eq:Nisomega}
k_{\mathbf{N}} = k_{\omega}.
\end{equation}
\end{proposition}
\begin{proof}
We first observe that $\omega$ defined by equation \eqref{eq:omegadef} does not depend
on the choice of conformal representative.  Indeed, if $\tilde g_{ab} = \phi^{q-2} g_{ab}$
for some conformal factor $\phi$ then $\tilde N = \phi^q N$ and $dV_{\tilde g} = \phi^q dV_g$
so
\begin{equation}
\frac{1}{2N} dV_g = \frac{1}{2\tilde N} dV_{\tilde g}.
\end{equation}
So to establish equation \eqref{eq:Nisomega} it suffices to work with a convenient
background metric.
Let $g_{ab}$ be the representative metric such that $dV_g = \omega$ (or equivalently such
that $N=1/2$).  

Suppose $\mathbf{u}=[g_{ab},u_{ab}]$ is a conformal tangent vector.
Since $N=1/2$, equation \eqref{eq:iotaN} then implies
\begin{equation}
k_{\mathbf{N}}(\mathbf{u}) = [g_{ab}, u_{ab}]^*.
\end{equation}
On the other hand, since $dV_g = \omega$, equation \eqref{eq:ttstarlapse} implies
\begin{equation}
k_{\omega}(\mathbf{u}) = [g_{ab}, (dV_g/\omega) u_{ab}]^* = [g_{ab},u_{ab}]^*
\end{equation}
as well. Hence $k_{\mathbf{N}}(\mathbf{u})=k_{\omega}(\mathbf{u})$.
\end{proof}

From Propositions \ref{prop:TvsTstar} and \ref{prop:Nisomega} it follows that
each $k_{\mathbf N}$ is a bijection onto the space of smooth cotangent vectors
and admits an inverse $k^{-1}_{\mathbf N}$.  It then follows from equation \eqref{eq:iotaN} 
that for any smooth conformal cotangent vector $\mathbf A=[g_{ab}, A_{ab}]^*$,
\begin{equation}\label{eq:iotaNinv}
k^{-1}_{\mathbf N} (\mathbf A) = [g_{ab}, 2N A_{ab}].
\end{equation}

From Proposition \ref{prop:Nisomega} we can translate Proposition \ref{prop:yorksplitconf}
in terms of densitized lapses.

\begin{proposition}\label{prop:yorksplitlapse}
Let $\mathbf{g}\in\calC$ and let $\mathbf{A}\in T^*_\mathbf{g} \calC$.  Given a choice of
densitized lapse $\mathbf{N}$
there is a conformal geometric momentum $\bfsigma$ and a vector field $W^a$ such that
\begin{equation}\label{eq:yorksplitlapse}
\mathbf{A} = \bfsigma + k_{\mathbf{N}}( \ck_\mathbf{g} W),
\end{equation}
The decomposition is unique up the the addition of a conformal Killing field to $W^a$.

Moreover, if $g_{ab}$ is an arbitrary representative of $\mathbf{g}$, and
if $A_{ab}$, $\sigma_{ab}$ and $N$ are the representatives of 
$\mathbf{A}$, $\bfsigma$ and $\mathbf{N}$ with respect to $g_{ab}$, then
\begin{equation}\label{eq:Alapse}
A_{ab} = \sigma_{ab} + \frac{1}{2N}(\ck_{g} W)_{ab}.
\end{equation}
\end{proposition}
\begin{proof}
Equation \eqref{eq:yorksplitlapse} is immediate from Propositions \ref{prop:yorksplitconf}
and \ref{prop:Nisomega} and it remains to establish equation \eqref{eq:Alapse}.

Starting from equation \eqref{eq:yorksplitlapse}, let $g_{ab}$ be a representative
of $\mathbf{g}$, and let $A_{ab}$, $\sigma_{ab}$ and $N$ be the representatives 
of $\mathbf{A}$, $\bfsigma$, and $\mathbf{N}$ with respect to $g_{ab}$.
By definition
\begin{equation}
\ck_\mathbf{g} W = [g_{ab}, (\ck_{g} W)_{ab}]
\end{equation}
and hence equation \eqref{eq:iotaN} implies
\begin{equation}
k_{\mathbf{N}}( \ck_\mathbf{g} W) = [g_{ab}, (1/2N) (\ck_g W)_{ab}]^*.
\end{equation}
So equation \eqref{eq:yorksplitlapse} reads
\begin{equation}  
[g_{ab}, A_{ab}]^* = [g_{ab},\sigma_{ab}]^* + [g_{ab}, (2N)^{-1} (\ck_g W)_{ab}]^*
\end{equation}
which establishes equation \eqref{eq:Alapse}.
\end{proof}

Although Propositions \ref{prop:yorksplitconf} and \ref{prop:yorksplitlapse}
express the same fact, equation \eqref{eq:Alapse} from Proposition \ref{prop:yorksplitlapse}
is more flexible than its counterpart equation \eqref{eq:Asplitconf}
from Proposition \ref{prop:yorksplitconf}.
Equation \eqref{eq:Alapse} is written with respect to an arbitrary
background metric whereas equation \eqref{eq:Asplitconf} is written
with respect to a single background metric (the one where $dV_g=\omega$).

From Definition \ref{def:proj} we have volume-form dependent projections
$P_{\omega}$ from $T^*_\mathbf{g} \calC$ to the subspace of conformal geometric momenta.
We similarly define densitized-lapse-dependent projections $P_{\mathbf{N}}$ by
\begin{equation}
P_{\mathbf{N}}(\mathbf{A})=\bfsigma
\end{equation}
where $\bfsigma$ is the unique conformal geometric momentum from equation \eqref{eq:yorksplitlapse}.  
Following the construction of Definition \ref{def:j} we also have densitized-lapse-dependent
identifications $j_{\mathbf N}:T_{\mathbf g}(\calC/\calD_0)\ra T^*_{\mathbf g}(\calC/\calD_0)$
defined by
\begin{equation}\label{eq:jN}
j_{\mathbf N}(\mathbf u + \Im \ck_{\mathbf g})  = P_{\mathbf N}(k_{\mathbf N}(\mathbf u)),
\end{equation}
and the analogue of the commutative-diagram \eqref{eq:ij} holds as well.  Indeed, all of
these objects are obtained simply by replacing $\mathbf N$ with the volume form $\omega$
defined in equation \eqref{eq:omegadef}.

Data for the CTS-H data method 
consists of a conformal class $\mathbf g$, a conformal geometric momentum $\bfsigma$,
a mean curvature $\tau$, and a densitized lapse $\mathbf N$ and we seek a solution $(\ovl g_{ab},\ovl K_{ab})$
of the constraints such that
\ifjournal
\numparts
\begin{eqnarray}
\customlabel{eq:CTSHparams}{\arabic{section}.\arabic{eqnval}}
[\ovl g_{ab}]&=\mathbf{g}\label{eq:CTSHcf}, \\
P_{\mathbf{N}}( [\ovl g_{ab}, \ovl K_{ab}]^*) &= \bfsigma,\label{eq:CTSHsigma} \\
\ovl g^{ab} \ovl K_{ab} &=\tau.
\end{eqnarray}
\endnumparts
\else
\begin{subequations}\label{eq:CTSHparams}
\begin{alignat}{2}
[\ovl g_{ab}]&=\mathbf{g}\label{eq:CTSHcf}, \\
P_{\mathbf{N}}( [\ovl g_{ab}, \ovl K_{ab}]^*) &= \bfsigma,\label{eq:CTSHsigma} \\
\ovl g^{ab} \ovl K_{ab} &=\tau.
\end{alignat}
\end{subequations}\fi
To formulate this problem in terms of a PDE, 
let $g_{ab}$ be an arbitrary representative of $\mathbf{g}$. Suppose
$(\ovl g_{ab},\ovl K_{ab})$ is a metric and second fundamental form
with $[\ovl g_{ab}]=[g_{ab}]$, so $\ovl g_{ab} = \phi^{q-2}g_{ab}$
for some conformal factor $\phi$. Let $\ovl A_{ab}$ be the trace-free
part of $\ovl K_{ab}$, and let
$A_{ab}$ and $\sigma_{ab}$ be the representatives of
$[\ovl g_{ab}, \ovl A_{ab}]^*$ and $\bfsigma$ with respect to
$g_{ab}$, so $\ovl A_{ab} = \phi^{-2} A_{ab}$.  From
the definition of $P_{\mathbf N}$ and equation \eqref{eq:yorksplitlapse}
we see that equation \eqref{eq:CTSHsigma} is equivalent to
\begin{equation}
A_{ab} = \sigma_{ab} + \frac{1}{2N} (\ck_{g} W)_{ab}
\end{equation}
for some vector field $W^a$.  Thus equations \eqref{eq:CTSHparams}
can be written in terms of the background metric $g_{ab}$ as
\ifjournal
\numparts
\begin{eqnarray}
\customlabel{eq:CTSHparams-coords}{\arabic{section}.\arabic{eqnval}}
\ovl g_{ab} &= \phi^{q-2} g_{ab}\\
\ovl K_{ab} &= \phi^{-2}\left( \sigma_{ab} + \frac{1}{2N}(\ck_{g}W)_{ab}\right) + \frac{\tau}{n}\ovl g_{ab}
\end{eqnarray}
\endnumparts
\else
\begin{subequations}\label{eq:CTSHparams-coords}
\begin{alignat}{2}
\ovl g_{ab} &= \phi^{q-2} g_{ab}\\
\ovl K_{ab} &= \phi^{-2}\left( \sigma_{ab} + \frac{1}{2N}(\ck_{g}W)_{ab}\right) + \frac{\tau}{n}\ovl g_{ab}
\end{alignat}
\end{subequations}\fi
for some conformal factor $\phi$ and vector field $W^a$.

Substituting equations \eqref{eq:CTSHparams-coords} into the constraint equations
we see that $(\ovl g_{ab}, \ovl K_{ab})$
solve the constraint equations if and only if $(\phi, W^a)$ satisfy
the CTS-H equations
\ifjournal
\numparts
\begin{eqnarray}
\customlabel{eq:ctsh}{\arabic{section}.\arabic{eqnval}}
\hskip-15ex
-2\kappa q \Lap_g \phi + R_g \phi -\left|\sigma + \frac{1}{2N}\ck_g W\right|^2_g\phi^{-q-1} + \kappa \tau^2\phi^{q-1} = 0
&\hskip 1 ex\text{\small[CTS-H Hamiltonian constraint]}
\\
\div_g \frac{1}{2N} \ck W = \kappa\phi^q d\tau.
&\hskip 1ex\text{\small[CTS-H momentum constraint]}
\end{eqnarray}
\endnumparts
\else
\begin{subequations}\label{eq:ctsh}
\begin{alignat}{2}
-2\kappa q \Lap_g \phi + R_g \phi -\left|\sigma + \frac{1}{2N}\ck_g W\right|^2_g\phi^{-q-1} + \kappa \tau^2\phi^{q-1} &= 0
&\qquad&\text{\small[CTS-H Hamiltonian constraint]}
\\
\div_g \frac{1}{2N} \ck W &= \kappa\phi^q d\tau.
&\qquad&\text{\small[CTS-H momentum constraint]}
\end{alignat}
\end{subequations}\fi
These are equivalent to the equations that appear in \cite{Pfeiffer:2003ka},
with differences appearing because we treat $\sigma_{ab}$ as the representative
of a conformal geometric momentum $\bfsigma$ that is freely specified whereas
\cite{Pfeiffer:2003ka} treats $\sigma_{ab}$ as something to be extracted 
as a TT component of a freely-specified source tensor $C_{ab}$. 

We summarize the previous discussion with the following proposition 
(noting that CTS-H data and representative data are defined analogously to their CTS-L counterparts).

% We call a tuple $(\mathbf g, \mathbf u, \tau, \mathbf N)$ \textbf{CTS-H data},
% and define \textbf{representatives} analogously to representatives of CTS-L data.

\begin{proposition}[The CTS-H Method]\label{prop:CTSH}
Let $(\mathbf g, \bfsigma, \tau, \mathbf N)$ be CTS-H data, and let
$(g_{ab}, \sigma_{ab}, \tau, N)$ be an arbitrary representative
of this data.

If $\phi$ and $W^a$ solve the CTS-H equations \eqref{eq:ctsh}
then $(\ovl g_{ab}, \ovl K_{ab})$ defined by equations \eqref{eq:CTSHparams-coords}
satisfy the constraint equations \eqref{eq:constraints} and satisfy
\ifjournal
\numparts
\begin{eqnarray}
\customlabel{eq:CTSHparams2}{\arabic{section}.\arabic{eqnval}}
\mathbf{g}&=[\ovl g_{ab}]\label{eq:CTSHcf2}, \\
\bfsigma &= P_{\mathbf{N}}( [\ovl g_{ab}, \ovl K_{ab}]^*),\quad\text{and}\label{eq:CTSHsigma2}\\
\tau &= \ovl g^{ab} \ovl K_{ab}.
\end{eqnarray}
\endnumparts
\else
\begin{subequations}\label{eq:CTSHparams2}
\begin{alignat}{2}
\mathbf{g}&=[\ovl g_{ab}]\label{eq:CTSHcf2}, \\
\bfsigma &= P_{\mathbf{N}}( [\ovl g_{ab}, \ovl K_{ab}]^*),\quad\text{and}\label{eq:CTSHsigma2}\\
\tau &= \ovl g^{ab} \ovl K_{ab}.
\end{alignat}
\end{subequations}\fi

Conversely, suppose $(\ovl g_{ab}, \ovl K_{ab})$ solve the constraint equations
and satisfy conditions \eqref{eq:CTSHparams2}. Then there exist
a conformal factor $\phi$ and a vector field $W^a$, unique up to
addition of a conformal Killing field to $W^a$, such that
the decomposition \eqref{eq:CTSHparams-coords} holds and the CTS-H equations
\eqref{eq:ctsh} are satisfied.
\end{proposition}

The CTS-H method is conformally covariant; the proof
is analogous to that of Proposition \ref{prop:CMCinv} and is omitted.
\begin{proposition}
Let $(g_{ab}, \sigma_{ab}, \tau, N)$ be representative CTS-H data, let $\psi$ be a smooth positive
function, and let $\tilde g_{ab}=\psi^{q-2}g_{ab}$, $\tilde \sigma_{ab} = \phi^{-2} \sigma_{ab}$,
and $\tilde N = \psi^q N$.  Then $(\phi,W)$ solves the CTS-H equations \eqref{eq:ctsh} with respect to 
the data $(g_{ab}, \sigma_{ab}, \tau, N)$ if and only if $(\psi^{-q}\phi,W)$ solve
the CTS-H equations with respect to $(\tilde g_{ab},\tilde  \sigma_{ab}, \tau,\tilde  N)$
and both yield the same solution $(\ovl g_{ab}, \ovl K_{ab})$ of the constraint equations.
\end{proposition}

\section{Equivalence of the Methods}\label{sec:same}

We now show that the 1974, CTS-L, and CTS-H parameterizations are all the same
by demonstrating how to translate between the parameters for these
methods such that the corresponding solutions of the constraints are the same.

Starting with the CTS-L and CTS-H methods,
the parameters $\mathbf{g}$, $\tau$ and $\mathbf{N}$ retain
their roles and are fixed when moving between the two methods and we
need a way to map back and forth between the velocity/momentum
parameters.  The momentum parameter from the CTS-H method is
an element of $T^*_\mathbf g (\calC/\calD_0)$, and we saw at the
end of Section \ref{sec:CTS} that the true velocity parameter
for the CTS method is a conformal geometric velocity 
$\mathbf u + \Im \ck_{\mathbf g}\in T_\mathbf g (\calC/\calD_0)$.
So a natural candidate for the identification is the map $j_{\mathbf N}$
defined in equation \eqref{eq:jN}, and this is the correct choice.

\begin{proposition}\label{prop:CTSL-CTSH}
The solutions of the constraint equations generated by 
geometric CTS-L data $(\mathbf{g},\mathbf{u}+\Im \ck_{\mathbf g},\tau,\mathbf{N})$ 
and the solutions generated by 
CTS-H data $(\mathbf{g},\bfsigma,\tau,\mathbf{N})$
coincide if and only if
\begin{equation}\label{eq:utosig}
\bfsigma = j_{\mathbf N}(\mathbf u + \Im \ck_{\mathbf g}).
\end{equation}

In terms of a representative metric $g_{ab}\in \mathbf g$,
representative CTS-L data $(g_{ab}, u_{ab}, \tau, N)$ and
representative CTS-H data $(g_{ab}, \sigma_{ab}, \tau, N)$
generate the same solutions if and only if there is a vector field $X^a$ such
that
\begin{equation}
u_{ab} = 2N \sigma_{ab} + (\ck_{g} X)_{ab}.
\end{equation}
\end{proposition}
\begin{proof}
Suppose $(\ovl g_{ab}, \ovl K_{ab})$ is a solution of the constraints generated by
CTS-L data $(\mathbf{g},\mathbf{u}+\Im \ck_{\mathbf g},\tau,\mathbf{N})$.
Then there exists a vector field $X^a$ such that conditions \eqref{eq:CTS-Lconds} 
hold.  In particular,
\begin{equation}
[\ovl g_{ab},\dot{\ovl g}_{ab}] = \mathbf u
\end{equation}
where $\dot{\ovl g}_{ab} = 2 \ovl N\, \ovl K_{ab} + \calL_{X} \ovl g_{ab}$.
Hence
\begin{equation}
[\ovl g_{ab},2 \ovl N\, K_{ab}] \in \mathbf u + \Im \ck_{\mathbf g}
\end{equation}
and therefore from equation \eqref{eq:jdef} we have
\begin{equation}
j_{\mathbf N}(\mathbf u + \Im \ck_{\mathbf g}) = P_{\mathbf N}(k_{\mathbf N}([\ovl g_{ab},2 \ovl N\ovl K_{ab}])).
\end{equation}
Equation \eqref{eq:iotaN} implies
\begin{equation}
k_{\mathbf N}([\ovl g_{ab},2 \ovl N, \ovl K_{ab}]) =[\ovl g_{ab}, \ovl K_{ab}]^*
\end{equation}
and thus
\begin{equation}
P_{\mathbf N}([\ovl g_{ab}, \ovl K_{ab}]^*) = j_{\mathbf N}(\mathbf u + \Im \ck_{\mathbf g}).
\end{equation}
Defining $\bfsigma=P_{\mathbf N}([\ovl g_{ab}, \ovl K_{ab}]^*)$, the solution $(\ovl g_{ab},\ovl K_{ab})$
then satisfies conditions \eqref{eq:CTSHparams2} and therefore is generated by 
CTS-H data $(\mathbf g, \bfsigma, \tau, \mathbf N)$.  

The previous discussion is reversible
and therefore if a solution of the constraints is generated by CTS-H data $(\mathbf g, \bfsigma, \tau, \mathbf N)$, 
then it is generated by geometric CTS-L data 
$(\mathbf g, \mathbf u + \Im \ck_{\mathbf g}, \tau, \mathbf N)$
where
\begin{equation}\label{eq:sigtou}
\mathbf u + \Im \ck_\mathbf g = j_{\mathbf N}^{-1}(\bfsigma).
\end{equation}

To reformulate equation \eqref{eq:sigtou} in terms of a background metric $g_{ab}\in \mathbf g$,
let $u_{ab}$, $\sigma_{ab}$ and $N$ be the representatives
of $\mathbf u$, $\bfsigma$ and $\mathbf N$ with respect to $g_{ab}$.  Since
\begin{equation}
j_{\mathbf N}^{-1}(\bfsigma) = k_{\mathbf N}^{-1}(\bfsigma) + \Im \ck_\mathbf g,
\end{equation}
and since
\begin{equation}
k_{\mathbf N}^{-1}(\bfsigma) = [g_{ab}, 2N\sigma_{ab}],
\end{equation}
equation \eqref{eq:sigtou} is equivalent to
\begin{equation}
u_{ab} = 2N \sigma_{ab} + (\ck_{g} X)_{ab}
\end{equation}
for some vector field $X^a$, where $u_{ab}$ is the representative
of $\mathbf u$ with respect to $u_{ab}$.
\end{proof}

The equivalence of the 1974 method and the CTS-H method
is a consequence of the equivalences of the projections $P_\omega$ for the 1974 method and
the projections $P_{\mathbf N}$ of the CTS-H method, where we translate between
volume forms and densitized lapses via equation \eqref{eq:omegadef}.

\begin{proposition}\label{prop:1974isCTS}
Let $\mathbf g\in\calC$, $\bfsigma\in T^*_{\mathbf g}(\calC/\calD_0)$ and $\tau\in C^\infty(M)$.
Suppose $\mathbf N\in\calN_{\mathbf g}$ is a densitized lapse and $\omega$ is a volume form
that satisfy equation \eqref{eq:omegadef}.
Then the set of solutions of the constraints generated by CTS-H data 
$(\mathbf g, \bfsigma, \tau, \mathbf N)$ is the same as the set of solutions generated by
1974 data $(\mathbf g, \bfsigma, \tau, \omega)$.
\end{proposition}
\begin{proof}
Suppose $(\ovl g_{ab}, \ovl K_{ab})$ is a solution of the constraints generated by CTS-H data
$(\mathbf g, \bfsigma, \tau, \mathbf N)$.  Then $(\ovl g_{ab}, \ovl K_{ab})$ satisfy
conditions \eqref{eq:CTSHparams2}, and in particular
\begin{equation}
P_{\mathbf N}([\ovl g_{ab}, \ovl K_{ab}]^*) = \bfsigma.
\end{equation}
But $P_{\mathbf N}=P_{\omega}$ where $\omega$ is defined in equation \eqref{eq:omegadef},
so the solution $(\ovl g_{ab}, \ovl K_{ab})$ satisfies conditions \eqref{eq:1974params2}
as well and is generated by 1974 data $(\mathbf g, \bfsigma, \tau, \omega)$.
The converse is proved similarly.
\end{proof}

Proposition \ref{prop:1974isCTS} admits the following reformulation 
in terms of 1974 representative data, where the volume form $\omega$
is determined implicitly by the background metric.

\begin{proposition}\label{prop:1974rep-CTSH}
Let $(g_{ab}, \sigma_{ab}, \tau)$ be representative 1974 data
and let
\begin{equation}
\ifjournal\eqalign{
\mathbf{g} &= [g_{ab}]\\
\bfsigma   &= [g_{ab},\sigma_{ab}]\\
\mathbf{N} &= [g_{ab}, 1/2].	
}\else
\begin{aligned}
\mathbf{g} &= [g_{ab}]\\
\bfsigma   &= [g_{ab},\sigma_{ab}]\\
\mathbf{N} &= [g_{ab}, 1/2].
\end{aligned}\fi
\end{equation}
Then the set of solutions of the constraint equations generated by
the 1974 method for
$(g_{ab},\sigma_{ab},\tau)$
is the same as the set of solutions generated by the CTS-H method
for data $(\mathbf{g}, \bfsigma,\tau,\mathbf{N})$.  

Conversely, suppose $(\mathbf{g},\bfsigma,\tau,\mathbf{N})$ is a tuple of CTS-H data.
Let $g_{ab}$ be the unique element of $\mathbf{g}$ such that
\begin{equation}
[g_{ab},1/2] = \mathbf{N}
\end{equation}
and let $\sigma_{ab}$ be the unique TT tensor such that
\begin{equation}
[g_{ab},\sigma_{ab}] = \bfsigma.
\end{equation}
The set of solutions of the constraint equations generated by 
by the CTS-H method for $(\mathbf{g},\bfsigma,\tau,\mathbf{N})$ 
is the same as the set generated by 1974 conformal data $(g_{ab},\sigma_{ab},\tau)$.
\end{proposition}
\begin{proof}
Representative 1974 data $(g_{ab}, \sigma_{ab}, \tau)$ determine 1974 data
$(\mathbf g, \bfsigma, \tau, \omega)$ with 
$\mathbf g = [g_{ab}]$, $\bfsigma=[g_{ab}, \sigma_{ab}]^*$, and $\omega = dV_{g}$.
Proposition \ref{prop:1974isCTS} implies that the set of solutions generated by 1974 data 
$(\mathbf g, \bfsigma, \tau, \omega)$ is the same as the set of solutions
generated by CTS-H data $(\mathbf g, \bfsigma, \tau, \mathbf N)$ where the
representative of $\mathbf N$ with respect to $g_{ab}$ satisfies
equation \eqref{eq:omegadef}.  Since $\omega = dV_g$, equation \eqref{eq:omegadef}
implies $N=1/2$ and hence
$\mathbf{N} = [g_{ab}, 1/2]$.
This establishes the forward direction, and the converse is proved similarly.
\end{proof}

While the `coordinate-free' approach to expressing the conformal method parameters is 
helpful, applications frequently require working with representative data. Summarizing from
Propositions \ref{prop:CTSL-CTSH} and \ref{prop:1974rep-CTSH} we translate between
representative data as follows.

\begin{itemize}
\item\textbf{[1974 to CTS-H]} Start with 1974 data $(g_{ab}, \sigma_{ab}, \tau)$ and
adjoin a lapse $N=1/2$.  Use CTS-H data $(g_{ab}, \sigma_{ab}, \tau, 1/2)$,
or any conformally related CTS-H data.

\item\textbf{[CTS-H to 1974]}  Start with CTS-H data $(g_{ab}, \sigma_{ab}, \tau, N)$ and
let $\psi$ be the conformal factor satisfying $\psi^q N = (1/2)$.  Let
$\hat g_{ab} = \psi^{q-2}g_{ab}$ and $\hat \sigma_{ab} = \psi^{-2}\sigma_{ab}$, and use
1974 data $(\hat g_{ab}, \hat \sigma_{ab}, \tau)$.

\item\textbf{[CTS-H to CTS-L]} Start with CTS-H data $(g_{ab}, \sigma_{ab}, \tau, N)$ 
and select an arbitrary vector field $X^a$.  Let 
\begin{equation}
u_{ab} = 2N\sigma_{ab} + (\ck_{g} X)_{ab},
\end{equation}
and use CTS-L data $(g_{ab}, u_{ab}, \tau, N)$ or any
conformally related CTS-L data.

\item\textbf{[CTS-L to CTS-H]} Start with CTS-L data $(g_{ab}, u_{ab}, \tau, N)$ 
and let $\sigma_{ab}$ be the unique transverse traceless tensor with
\begin{equation}
u_{ab} = 2N\sigma_{ab} + (\ck_{g} Y)_{ab}
\end{equation}
for some vector field $Y^a$, as given by Proposition \ref{prop:yorksplitlapse}.
Use CTS-H data $(g_{ab}, \sigma_{ab}, \tau, N)$ or any conformally
related CTS-H data.
\end{itemize}

\section{Applications}\label{sec:apps}

In this section we strengthen two previous results concerning the 1974 conformal method on
compact manifolds by using the correspondence between the 1974 method and the CTS-H method.

\subsection{Near-CMC Existence/Uniqueness}
The main theorem from \cite{Isenberg:1996fi}, when restricted to smooth tensors, can be
phrased as follows.
\begin{theorem}\label{thm:im}
Let $M^3$ be a compact 3-manifold.  Suppose $g_{ab}$ is a smooth metric on $M$ that has constant scalar curvature
equal to -1 and that admits only the trivial conformal Killing field, and 
suppose $\sigma_{ab}$ is an arbitrary transverse-traceless
tensor with respect to $g_{ab}$. Then there is an open set $T_{g,\sigma}$ 
of nowhere-vanishing mean curvatures such that
every nonzero constant mean curvature belongs to $T_{g,\sigma}$,
and such that for every $\tau \in T_{g,\sigma}$ the 1974 conformally-parameterized 
constraint equations \eqref{eq:LCBY} for
the representative 1974 data $(g_{ab}, \sigma_{ab}, \tau)$ have a unique
solution.
\end{theorem}
The set $T_{g,\sigma}$ in Theorem \ref{thm:im} is defined by
\begin{equation}
\frac{\max{|d \tau|_g}}{\min |\tau|} \qquad\text{and}\qquad |d\tau|_g
\end{equation}
being sufficiently small, so Theorem \ref{thm:im} is a near-CMC existence and uniqueness result.
It is remarked in \cite{Isenberg:1996fi} that the proof of Theorem \ref{thm:im} could be carried out
under the more general hypothesis $R_g<0$ everywhere, but that the authors were unable to
extend it to the most natural generalization that $g_{ab}$ is Yamabe negative.  We show
here that such an extension is possible. 

In coordinate-free language, Theorem \ref{thm:im} can be phrased as follows.
\begin{theorem}
Let $M^3$ be a compact 3-manifold.  Suppose $\mathbf{g}$ is a Yamabe-negative conformal class
on $M$ admitting only the trivial conformal Killing field, and suppose $\bfsigma\in T^*_\mathbf g\calC$.
Let $\omega$ be the volume form of the unique representative $g_{ab}\in\mathbf{g}$
that satisfies $R_{g}=-1$.  Then there is an open set $T_{\mathbf g,\bfsigma}$ 
of nowhere-vanishing mean curvatures such that
every nonzero constant mean curvature belongs to $T_{g,\sigma}$,
and such that for every $\tau \in T_{g,\sigma}$ the
1974 conformal data $(\mathbf g, \bfsigma, \tau, \omega)$
determines a unique solution of the constraint equations.
\end{theorem}
In this language, the central restriction of the theorem is the choice of a single volume form $\omega$.
We wish to eliminate this restriction,
and we use the fact that the choice of volume form for the 1974 method corresponds to the choice of densitized lapse for the CTS-H method.
So we will consider the CTS-H equations
\begin{equation}\label{eq:ctsh3}
\ifjournal\eqalign{
-2\kappa q \Lap_g \phi + R_g \phi -\left|\sigma + \frac{1}{2N}\ck_g W\right|^2_g\phi^{-q-1} + \kappa \tau^2 &= 0\\
\\
\div_g \frac{1}{2N} \ck W &= \phi^q d\tau	
}\else
\begin{aligned}
-2\kappa q \Lap_g \phi + R_g \phi -\left|\sigma + \frac{1}{2N}\ck_g W\right|^2_g\phi^{-q-1} + \kappa \tau^2 &= 0\\
\\
\div_g \frac{1}{2N} \ck W &= \phi^q d\tau
\end{aligned}\fi
\end{equation}
where $g_{ab}$ is the unique representative with $R_g=-1$ and $N$ is an arbitrary lapse. 

The equations
considered by Theorem \ref{thm:im} are exactly equations \eqref{eq:ctsh3} with $N=1/2$,
so we need to consider the impact of an arbitrary choice of $N$ in equations 
\eqref{eq:ctsh3} on the rather technical proof of Theorem \ref{thm:im}.
In effect, this amounts to replacing $\ck_g$ with $1/(2N)\ck_g$ wherever it appears in the proof, and
there are facts concerning the vector Laplacian
$\Lap_{\ck_g} = \div_g \ck_g$ that need to be revisited for the operator
\begin{equation}
\Lap_{\ck_g,N}=\div_g (2N)^{-1}\ck_g.
\end{equation}
\begin{proposition}\label{prop:divck}
Let $g_{ab}$ and $N$ be a smooth metric and positive smooth function on $M$.
The operator $\Lap_{\ck_g,N}$ is linear, elliptic, and
self-adjoint with respect to $g_{ab}$. If $g_{ab}$ has no conformal Killing fields
then $\Lap_{\ck_g,N}$ has trivial kernel.  Regardless of whether $g_{ab}$ has
conformal Killing fields, there is a constant $c_{g,N}$ such that
if $W^a$ and $\eta_a$ satisfy
\begin{equation}
\Lap_{\ck_g,N} W = \eta
\end{equation}
then
\begin{equation}\label{eq:divckNest}
\left| \frac{1}{2N} \ck_g W\right|_g \le c_{g,N} |\eta|_g.
\end{equation}
\end{proposition}
\begin{proof}
The fact that $\Lap_{\ck_g,N}$ is linear, elliptic, and self-adjoint is obvious, and
an integration by parts argument shows that its kernel consists of conformal Killing fields,
so it remains to establish inequality \eqref{eq:divckNest}.

Let $\psi$ be the unique positive function with $\psi^q = (2N)^{-1}$
and let $\tilde g_{ab} = \psi^{q-2} g_{ab}$.  Then $\ck_{\tilde g} = \psi^{q-2} \ck_g$ and
$\div_{\tilde g} = \psi^{2-2q} \div_g \psi^2$.  Hence
\begin{equation}
\Lap_{\ck_{\tilde g}} = \psi^{-q} \div_g \psi^q \ck_g = \psi^{-q} \div_g (2N)^{-1}\ck_g 
= \psi^{-q} \Lap_{\ck_g,N}.
\end{equation}
Now suppose $W^a$ and $\eta_a$ satisfy
$\Lap_{\ck_g,N} W = \eta$, so 
\begin{equation}
\Lap_{\ck_{\tilde g}} W = \psi^{q}\eta.
\end{equation}
From \cite{Isenberg:2004jd} Lemma 1 concerning the standard vector Laplacian
we know that there is a constant $\tilde c$, independent of $W^a$ and $\eta_a$, 
such that
\begin{equation}
\max |\ck_{\tilde g}W|_{\tilde g} \le \tilde c \max (\psi^q) \max|\eta|_{\tilde g}.
\end{equation}
Since the norms for $g$ and $\tilde g$ are comparable via constants depending
on $\min\psi$ and $\max\psi$, and 
since $\ck_{\tilde g} = \psi^{q-2}\ck_g$, inequality \eqref{eq:divckNest} now follows,
where the constant depends on $\hat c$, $\min \psi$ and $\max\psi$ (i.e., on $g$ and $N$).
\end{proof}

With Proposition \ref{prop:divck} in hand, the reader is now invited to walk through
the proof of Theorem \ref{thm:im}, as presented in \cite{Isenberg:1996fi},
to establish existence and uniqueness of equations \eqref{eq:ctsh3}.
The only interesting changes occur in establishing analogues of inequalities (38) and (58)
of that paper under the hypotheses that $\max(|d\tau|_g/|\tau|)$ and $|d\tau|_g$ are
sufficiently small. This is exactly where inequality \eqref{eq:divckNest} of Proposition
\ref{prop:divck} is used.  In coordinate-free language, one arrives at the following result.

\begin{theorem}\label{thm:nearCMC}
Let $M^3$ be a compact 3-manifold.  Let $\mathbf{g}$ be a Yamabe-negative conformal class
on $M$ admitting no conformal Killing fields, and let $\bfsigma\in T^*_\mathbf g\calC$ be arbitrary.
For any choice $\mathbf N$ of densitized lapse there is an open set 
$T_{\mathbf g,\bfsigma,\mathbf N}$ of 
of nowhere-vanishing mean curvatures such that
every nonzero constant mean curvature belongs to $T_{g,\sigma,\mathbf N}$, and
such that for every $\tau \in T_{\mathbf g,\bfsigma,\mathbf N}$ the
CTS-H data $(\mathbf g, \bfsigma, \tau, \mathbf{N})$
determines a unique solution of the constraint equations.
\end{theorem}

Using Proposition \ref{prop:1974rep-CTSH}, Theorem \ref{thm:nearCMC} then implies that
Theorem \ref{thm:im} holds with the restriction $R_g=-1$ replaced by the condition that $g_{ab}$
is Yamabe-negative.

Reference \cite{Allen:2008ef} contains results that are analogues of
Theorem \ref{thm:im} under the hypotheses
that $R_{g}\equiv 0$ or $R_{g}\equiv 8$.  We assert that using Proposition \ref{prop:divck}
one can repeat the exercise just undertaken for this paper as well to prove the following.

\begin{theorem}\label{thm:nearCMCpos}
Let $M^3$ be a compact 3-manifold.  Let $\mathbf{g}$ be a Yamabe-nonnegative conformal class
on $M$ admitting no conformal Killing fields, and let $\bfsigma\in T^*_\mathbf g\calC$ be
arbitrary (but not zero).
For any choice $\mathbf N$ of densitized lapse there is an open set 
$T_{\mathbf g,\bfsigma,\mathbf N}$ of nowhere-vanishing mean curvatures
such that
every nonzero constant mean curvature belongs to $T_{g,\sigma,\mathbf N}$,
and
such that for every $\tau \in T_{\mathbf g,\bfsigma,\mathbf N}$ the
CTS-H data $(\mathbf g, \bfsigma, \tau, \mathbf{N})$
determines a unique solution of the constraint equations.
\end{theorem}
Hence Theorem \ref{thm:im} also holds without any restriction whatsoever on the metric $g_{ab}$.

\subsection{Near-CMC Nonexistence}
Theorem \ref{thm:CMC} states that aside from some special cases,
there does not exists a solution for CMC data
$(\mathbf{g},\bfsigma,\tau)$ if $\mathbf g$ is Yamabe non-negative and $\bfsigma = 0$.
Reference \cite{Isenberg:2004jd} established the following two near-CMC analogues of this fact.

\begin{theorem}\label{thm:IOMa}
Let $M^3$ be a compact 3-manifold.  Suppose we have 1974 representative data $(g_{ab},\sigma_{ab},\tau)$ 
with $R_g\ge 0$ and $\sigma_{ab}\equiv 0$.  If $\tau = T+\rho$ for some nonzero constant $T$ and if
\begin{equation}
\frac{|d\rho|_g}{|T|}
\end{equation}
is sufficiently small, then the 1974 conformally parameterized constraint equations \eqref{eq:LCBY} do not admit a solution.
\end{theorem}
\begin{theorem}\label{thm:IOMb}
Let $M^3$ be a compact 3-manifold.  Suppose we have CTS-L representative data $(g_{ab},u_{ab},\tau,N)$
where $g$ is Yamabe non-negative and with $u_{ab}\equiv 0$.  
If $\tau = T+\rho$ for some nonzero constant $T$ and if
\begin{equation}\label{eq:iom}
\frac{|d\rho|_g}{|T|}
\end{equation}
is sufficiently small, then the CTS-L equations \eqref{eq:CTS} do not admit a solution.
\end{theorem}

Note that Theorem \ref{thm:IOMa} only applies to metrics with everywhere non-negative scalar curvature,
whereas Theorem \ref{thm:IOMb} only assumes the metric is Yamabe non-negative.  We now show that
Theorem \ref{thm:IOMa} can be strengthened to include the case that $g_{ab}$ is Yamabe non-negative.

Suppose $g_{ab}$ is Yamabe non-negative and that $\sigma_{ab}\equiv 0$.
Following the recipes at the end of Section \ref{sec:same}, if a solution of the
constraints exists for 1974 data $(g_{ab}, \sigma_{ab}, \tau)$ then it exists for
CTS-H data $(g_{ab}, \sigma_{ab}, \tau,N)$ where $N=1/2$.  And if it exists for this CTS-H data,
then it also exists for CTS-L data $(g_{ab}, 2N\sigma_{ab}, \tau, N)$.  That is, there
is a solution for CTS-L data $(g_{ab}, u_{ab}, \tau, 1/2)$ where $u_{ab}\equiv 0$.  Now Theorem
\ref{thm:IOMb} implies that if $\tau = T+\rho$ for some non-zero constant $T$, and if $|d\rho|_g/T$ is
sufficiently small, there is no solution for the CTS-L data $(g_{ab}, u_{ab}, \tau, 1/2)$ and
therefore no solution for the 1974 data $(g_{ab}, \sigma_{ab}, \tau)$.

For completeness we state the coordinate-free variation of this result and leave the proof as an exercise.
\begin{theorem}\label{thm:nogo}
Suppose $\mathbf g$ is a Yamabe-nonnegative conformal class.
\begin{enumerate}
\item  If $\omega$ is a volume form, there is an open set $U_{\mathbf g,\omega}$ of mean curvatures
that contains the non-zero constants such that 1974 data $(\mathbf g, \bfsigma, \tau, \omega)$
does not generate a solution of the constraints if $\bfsigma=0$ and $\tau\in U_{\mathbf g,\omega}$.
\item  If\, $\mathbf N$ is a densitized lapse, there is an open set $V_{\mathbf g,\mathbf N}$ of mean curvatures
that contains the non-zero constants such that CTS-H data $(\mathbf g, \bfsigma , \tau, \mathbf N)$
does not generate a solution of the constraints if $\bfsigma=0$ and $\tau\in V_{\mathbf g,\mathbf N}$.
\item  For the same set $V_{\mathbf g,\mathbf N}$ as in item 2, CTS-L data $(\mathbf g, \mathbf u , \tau, \mathbf N)$
does not generate a solution of the constraints if $\mathbf u\in\Im \ck_{\mathbf g}$ and $\tau\in V_{\mathbf g,\mathbf N}$.
\end{enumerate}
Moreover, the sets $U_{\mathbf g,\omega}$ and $V_{\mathbf g,\mathbf N}$ are the same if $\omega$ and $\mathbf N$ are related via
Proposition \ref{prop:Nisomega}.
\end{theorem}
Theorem \ref{thm:nogo} is not as specific as Theorem \ref{thm:IOMb} in defining the near-CMC condition because
we currently have a hazy understanding of what this set is. Expression
\eqref{eq:iom} is defined with respect to a particular representative metric, and the set $V_{\mathbf g,\mathbf N}$
can be thought of as taking a union of sets obtained from applying Theorem \ref{thm:IOMb} for each choice of background metric.
The maximal set $V_{\mathbf g, \mathbf N}$ for which Theorem \ref{thm:nogo} applies should be described 
in terms of $\mathbf g$ and $\mathbf N$ directly, and such a description is not yet understood.

\section{Conclusion}\label{sec:conclusion}
We have demonstrated in this paper that there is really only one conformal method.  The CMC conformal method
is a special case of the 1974 method, and the 1974 method has equivalent formulations in terms of the CTS-L
and CTS-H methods.  The parameters of the conformal method are:
\begin{enumerate}
\item a conformal class $\mathbf g$,
\item either a conformal geometric velocity $\mathbf u + \Im \ck_\mathbf g\in T_{\mathbf g}(\calC/\calD_0)$ 
or a conformal geometric momentum $\bfsigma\in T^*_{\mathbf g}(\calC/\calD_0)$,
\item a mean curvature $\tau$,
\item a choice of one member of a family of identifications of $T_\mathbf g\calC$ with $T_\mathbf g^*\calC$.
\end{enumerate}
The choice in item 4) can be made alternatively by selecting a volume form $\omega$ and using the map $k_\omega$ 
from Proposition \ref{prop:TvsTstar}, or by selecting a densitized lapse $\mathbf N$ and using the map
$k_{\mathbf N}$ from Definition \ref{def:iotaN}.  Proposition \ref{prop:Nisomega} shows how to
convert back and forth between $k_{\mathbf N}$ and $k_\omega$, so these are equivalent ways of expressing
the same choice. 
The CTS-L and CTS-H methods make choice 4) explicitly via $\mathbf N=[g_{ab},N]$, whereas 
the 1974 method makes the choice implicitly via $\omega = dV_g$.

The choice of $k_{\mathbf N}$ determines a related identification $j_{\mathbf N}$, given by 
equation \eqref{eq:jN}, between
conformal geometric velocities and momenta. This identification allows one to map back and forth
in item 2) between velocities and momenta: the CTS-L method uses velocities, whereas the 1974 and CTS-H 
methods use momenta, but by using $j_\mathbf N$ these are equivalent ways of expressing the same parameter.

The unifying theme of this paper is the need to clearly distinguish between tangent and cotangent vectors 
in the conformal method, and that there is a choice in the conformal method of how to identify these objects.  This
leads to the question of what this choice corresponds to (physically or otherwise). In fact, these identifications
arise as Legendre transformations in the $n$+1 formulation of gravity when using a densitized 
lapse.  We will return to this point and related results in forthcoming work.

Because the conformal methods are equivalent, a theorem proved for one method determines analogous theorems proved 
for the other methods.  From a practical point of view, however, the CTS-H method seems most expedient to
work with.  Given a choice of CTS-H data $(\mathbf g, \bfsigma, \tau, \mathbf N)$, the
CTS-H equations \eqref{eq:ctsh} can be expressed with respect to any representative background metric,
whereas given 1974 data $(\mathbf g, \bfsigma, \tau, \omega)$, the 1974 conformally parameterized constraint equations \eqref{eq:LCBY}
are written with respect to the unique background metric with $dV_g=\omega$.  This lack of flexibility
led to unnecessary restrictions in the past for theorems proved for the 1974 method, and we saw in Section
\ref{sec:same} how these restrictions can be overcome.  In principle one could express the 1974 method with
respect to an arbitrary background metric, and doing so must lead to the CTS-H equations 
after converting the volume form into a densitized lapse.  So there is little reason to prefer the 
1974 conformally parameterized constraint equations.
The only mild additional difficulty in working
with the CTS-H equations comes from working with the generalized vector Laplacians 
$\Lap_{g,N} = \div_g (2N)^{-1}\ck_g$ instead of the standard vector Laplacian
$\Lap_{g} = \div_g \ck_g$.  But the operators are very closely related and proofs
in the generalized case can typically be obtained by trivially modifying proofs for the standard case.
Moreover, as seen in Proposition \ref{prop:divck},
one can sometimes obtain results for the generalized operators as a corollary of a known results
for the standard vector Laplacian without revisiting the steps of the original proof.

The case for using the CTS-H equations over the CTS-L equations is not especially strong, but there are some
advantages.  The velocity parameter $u_{ab}$ in the CTS-L method is really a representative of the
whole subspace $u_{ab}+\Im \ck_g$, and this makes uniqueness statements a little more cumbersome for the CTS-L method.  
Moreover,
the CTS-H equations \eqref{eq:ctsh} are a little simpler than the CTS-L equations \eqref{eq:CTS}, and
are more familiar for researchers accustomed to working with the 1974 method: simply prepend a $1/(2N)$ in 
front of every conformal Killing operator and proceed as before.

\section*{Acknowledgment}
% \begin{acknowledgment}
This work was supported by NSF grant 0932078 000 while I was a resident at 
the Mathematical Sciences Research Institute in Berkeley, California,
and was additionally supported by NSF grant 1263544.
I would like to thank Jim Isenberg for many helpful conversations and especially
for pointing out the scalar curvature restrictions in references 
\cite{Isenberg:1996fi} and \cite{Allen:2008ef}.
% \end{acknowledgment}

\bibliographystyle{amsalpha-abbrv}
\bibliography{ConfAllSame,ConfAllSame-manual}

\end{document}